\documentclass[10pt,reqno]{amsart}

\usepackage[T1]{fontenc}
\usepackage{amsmath,amssymb,amsthm}
\usepackage{mathtools}
\usepackage{mathrsfs}
\usepackage{microtype}
\usepackage{graphicx}
\usepackage{hyperref}
\usepackage[nameinlink,capitalize,noabbrev]{cleveref}

\hypersetup{
  hidelinks,
  pdftitle={Distance Energies and Negative Type of Flat Tori},
  pdfauthor={Ye Zhou}
}

\numberwithin{equation}{section}

\newtheorem{theorem}{Theorem}[section]
\newtheorem{proposition}[theorem]{Proposition}
\newtheorem{lemma}[theorem]{Lemma}
\newtheorem{corollary}[theorem]{Corollary}

\theoremstyle{definition}

\theoremstyle{remark}
\newtheorem{remark}[theorem]{Remark}

\title[Distance energies and negative type]
{Distance Energies and Negative Type of Flat Tori}

\author[Y. Zhou]{Ye Zhou}
\address{Independent Researcher, Kunshan, Jiangsu, China}
\email{ye.zhou.horizon@gmail.com}


\date{}

\begin{document}

\begin{abstract}
Let \(T_\Lambda=\mathbb R^d/\Lambda\), \(d\geq2\), be a flat torus with quotient metric \(\rho_\Lambda\), and consider the distance energies \(I_\alpha(\mu)=\iint \rho_\Lambda(x,y)^\alpha\,d\mu(x)\,d\mu(y)\) of Borel probability measures \(\mu\). We prove a quantitative Fourier signature of the cut locus: for every Voronoi facet and every \(\alpha>0\), there is a sequence of dual-lattice frequencies approaching the facet normal along which the Fourier coefficients of \(\rho_\Lambda^\alpha\) are positive, with an explicit leading asymptotic determined by the facet. As direct consequences, Haar measure is not a local maximizer for any positive distance power, even among smooth densities, and every flat torus of dimension at least two has supremal negative type and generalized roundness zero. We then solve the global maximization problem for two classes of flat tori. On an orthogonal rectangular torus, the maximizers undergo a transition at \(\alpha=2\): for \(1\leq\alpha<2\) they are the translated uniform measures on the two-torsion subgroup; at \(\alpha=2\) all balanced couplings on translates of that subgroup are extremal; and for \(\alpha>2\) only equally weighted diametral pairs remain. On the regular hexagonal torus, the maximizers are precisely the uniform measures on translates of a distinguished cyclic subgroup of order three for every \(\alpha\geq1\).
\end{abstract}

\maketitle

\section{Introduction}
\label{sec:introduction}

Let \(\Lambda\subset\mathbb R^d\) be a full-rank lattice and let
\[
    T_\Lambda=\mathbb R^d/\Lambda
\]
be the associated flat torus, equipped with the quotient metric
\[
    \rho_\Lambda(x,y)
    =
    \min_{\lambda\in\Lambda}|x-y-\lambda|.
\]
For \(\alpha>0\) and \(\mu\in\mathcal P(T_\Lambda)\), we consider the distance energy
\begin{equation}
    I_\alpha(\mu)
    =
    \iint_{T_\Lambda\times T_\Lambda}
    \rho_\Lambda(x,y)^\alpha\,
    d\mu(x)\,d\mu(y).
    \label{eq:intro-distance-energy}
\end{equation}
We pursue two related questions. The first is spectral: how are the Fourier signs of \(\rho_\Lambda^\alpha\) controlled by the cut locus, and what do they imply for Haar measure and negative type? The second is global: in which geometries can the maximum of \(I_\alpha\) and all maximizing measures be determined exactly?

Translation invariance makes Fourier analysis natural. If \(k\) is a real even translation-invariant kernel and \(\xi\in\Lambda^*\setminus\{0\}\), then a trigonometric perturbation of normalized Haar measure satisfies
\[
\begin{aligned}
    &\mathcal E_k\!\left(
        \bigl(1+t\cos(2\pi\xi\cdot x)\bigr)m_\Lambda
    \right)
    -
    \mathcal E_k(m_\Lambda)\\
    &\hspace{35mm}
    =
    \frac{t^2}{2}\,\widehat k(\xi).
\end{aligned}
\]
Thus a single positive Fourier coefficient gives a direction in which Haar measure fails to maximize the energy.

The same sign has a metric consequence. Recall that a metric space \((X,\rho)\) has \(p\)-negative type if
\[
    \sum_{i,j=1}^N
    c_ic_j\,\rho(x_i,x_j)^p
    \leq0
\]
whenever
\[
    \sum_{i=1}^N c_i=0.
\]
A positive Fourier coefficient of \(\rho_\Lambda^p\) produces a zero-mass signed measure with positive quadratic form and hence, after atomic approximation, a finite violation of \(p\)-negative type. The Fourier signs of the distance kernel therefore govern both Haar instability and negative type.

The source of these signs is the cut locus. Inside the Voronoi cell of the origin, the distance is Euclidean; across a facet its gradient jumps. The second distributional derivative of a radial distance kernel therefore contains a singular measure supported on the facets. By choosing dual-lattice frequencies asymptotically normal to one facet and locking the Fourier phase on its affine span, this singular contribution can be isolated from the remaining facets.

Let \(V=V(\Lambda)\) be the Voronoi cell of the origin. A nonzero \(v\in\Lambda\) is Voronoi relevant if
\[
    F_v
    =
    \left\{
        x\in V:
        x\cdot v=\frac{|v|^2}{2}
    \right\}
\]
is a facet of \(V\). Our first result gives a precise Fourier signature of such a facet.

\begin{theorem}[Fourier signature of the cut locus]
\label{thm:intro-fourier-signature}
Let \(\Lambda\subset\mathbb R^d\) be a full-rank lattice with \(d\geq2\), let \(\alpha>0\), and let \(v\) be Voronoi relevant. There exist \(t_j\to\infty\) and \(\xi_j\in2\Lambda^*\) such that
\[
    |\xi_j-t_jv|\longrightarrow0
\]
and
\begin{equation}
    (2\pi|\xi_j|)^2
    \widehat{\rho_\Lambda^\alpha}(\xi_j)
    \longrightarrow
    \frac{\alpha|v|}{|V|}
    \int_{F_v}|x|^{\alpha-2}\,dS(x)
    >0.
    \label{eq:intro-fourier-asymptotic}
\end{equation}
In particular,
\[
    \widehat{\rho_\Lambda^\alpha}(\xi_j)>0
\]
for all sufficiently large \(j\).
\end{theorem}

The conclusion is directional: it produces positive subsequences near the normal of each Voronoi facet, rather than eventual positivity in all high-frequency directions. The leading term retains explicit geometric information about the selected facet. The argument is formulated more generally in \cref{thm:positive-fourier-tail} for kernels of the form
\(K_\Phi=\Phi(\rho_\Lambda)\), under the regularity hypotheses stated there and with \(\Phi'>0\); distance powers are the principal specialization used here.

The Fourier theorem has two immediate consequences at every positive exponent: a metric obstruction and an energy instability.

\begin{theorem}[Negative type of flat tori]
\label{thm:intro-negative-type}
Let \(\Lambda\subset\mathbb R^d\) be a full-rank lattice with \(d\geq2\). Then
\begin{equation}
    \wp(T_\Lambda)
    =
    \operatorname{gr}(T_\Lambda)
    =
    0.
    \label{eq:intro-negative-type-zero}
\end{equation}
Equivalently, \(T_\Lambda\) fails to have \(p\)-negative type for every \(p>0\).
\end{theorem}

Here \(\wp\) denotes supremal negative type and \(\operatorname{gr}\) generalized roundness. Their equality is due to Lennard, Tonge, and Weston \cite{LennardTongeWeston1997}. In the probability literature, the same supremum is called the \emph{fractional index}; see Istas \cite[\S2.6]{Istas2012}.

The value in \cref{eq:intro-negative-type-zero} is also accessible through geometric obstruction theory. Kokkendorff's classification rules out \(1\)-negative type for flat tori of dimension at least two \cite{Kokkendorff2007}. Venet developed closed-geodesic obstructions for manifold-indexed fractional Brownian fields \cite{Venet2016} and proved that every product cylinder \(S^1\times(0,\varepsilon)\) has fractional index zero, with standard product flat tori as an application \cite[Theorem~1 and Example~4.1]{Venet2019}. A shortest-vector reduction gives the corresponding cylinder obstruction for an arbitrary flat torus; see Remark~\ref{rem:venet-alternative-proof}. The Fourier approach complements these geometric obstructions by producing facet-wise positive high-frequency subsequences together with explicit leading constants. Related recent work of Doust and Weston gives upper bounds on supremal negative type for nontrivial \(\ell_2\)-products with a circle factor \cite[Theorem~4.6]{DoustWeston2026}.

The same positive Fourier modes yield a local energy consequence.

\begin{corollary}[Haar instability]
\label{cor:intro-haar-instability}
Under the hypotheses of \cref{thm:intro-negative-type}, for every \(\alpha>0\) and every \(\varepsilon>0\), there exists a smooth probability density \(r\) such that
\[
    \|r-1\|_{L^\infty(m_\Lambda)}<\varepsilon
\]
and
\[
    I_\alpha(r\,m_\Lambda)>I_\alpha(m_\Lambda).
\]
\end{corollary}

The general Fourier argument does not determine the global maximizers. For that problem the geometry of the lattice enters more strongly. We give complete answers for orthogonal rectangular tori and for the regular hexagonal torus.

Let
\[
    T_{\mathbf L}^d
    =
    \prod_{i=1}^d
    \mathbb R/(2L_i\mathbb Z),
    \qquad
    L_i>0,
\]
with \(d\geq2\), and put
\[
    D^2=\sum_{i=1}^dL_i^2,
    \qquad
    H_2=\prod_{i=1}^d\{0,L_i\}.
\]
For \(a\in T_{\mathbf L}^d\), let
\[
    \mu_a
    =
    2^{-d}
    \sum_{\varepsilon\in\{0,1\}^d}
    \delta_{a+(\varepsilon_1L_1,\ldots,\varepsilon_dL_d)}.
\]

\begin{theorem}[Rectangular tori]
\label{thm:intro-rectangular}
Let \(d\geq2\).

\begin{enumerate}
\item[\textup{(i)}]
If \(1\leq\alpha<2\), then
\[
    \max_{\mu\in\mathcal P(T_{\mathbf L}^d)}
    I_\alpha(\mu)
    =
    2^{-d}
    \sum_{\varepsilon\in\{0,1\}^d}
    \left(
        \sum_{i=1}^d\varepsilon_iL_i^2
    \right)^{\alpha/2},
\]
and the maximizers are precisely the measures \(\mu_a\).

\item[\textup{(ii)}]
If \(\alpha=2\), then
\[
    \max_{\mu\in\mathcal P(T_{\mathbf L}^d)}
    I_2(\mu)
    =
    \frac{D^2}{2}.
\]
The maximizers are exactly the probability measures supported on a translate of \(H_2\) whose coordinate marginals are balanced between the two antipodal coordinate values.

\item[\textup{(iii)}]
If \(\alpha>2\), then
\[
    \max_{\mu\in\mathcal P(T_{\mathbf L}^d)}
    I_\alpha(\mu)
    =
    \frac{D^\alpha}{2},
\]
and the maximizers are precisely the equally weighted diametral pairs.
\end{enumerate}
\end{theorem}

Thus \(\alpha=2\) is the only transition point in the range \(\alpha\geq1\). At the transition exponent itself the equality class is larger than in either adjacent regime: arbitrary correlations are allowed as long as the coordinate marginals are balanced.

For the equal-side standard product torus, the cases \(\alpha=2\) and \(\alpha>2\) are contained, up to rescaling, in the classification of Damelin and Nathe for the Riesz kernels
\[
    K_s(x,y)=\operatorname{sign}(s)\rho(x,y)^{-s}.
\]
Their cases \(s=-2\) and \(s<-2\) correspond to maximizing \(\rho^2\) and \(\rho^\alpha\) with \(\alpha>2\), respectively \cite[Theorem~5.1]{DamelinNathe2024}. Parts~\textup{(ii)} and~\textup{(iii)} of \cref{thm:intro-rectangular} extend these classifications to arbitrary rectangular side lengths, while part~\textup{(i)} treats the subquadratic regime \(1\leq\alpha<2\), which is not covered by that result.

For the standard square torus \(\mathbb T^2=\mathbb R^2/\mathbb Z^2\), part~\textup{(i)} gives an \(\alpha=1\) endpoint behavior different from that recorded in Dmitriy Bilyk's Problem~4.2 in \cite{Manskova2025}, where Haar measure is stated to be maximizing. If \(m\) denotes Haar measure and \(\nu_{H_2}\) normalized counting measure on \(\{0,\tfrac12\}^2\), then
\[
    I_1(m)
    =
    \frac{\sqrt2+\log(1+\sqrt2)}6
    <
    \frac{1+1/\sqrt2}{4}
    =
    I_1(\nu_{H_2}).
\]
(The first value is obtained by integrating \(\sqrt{x^2+y^2}\) over
\([-\tfrac12,\tfrac12]^2\).)
At \(\alpha=2\), the complete balanced-coupling equality class is already present in \cite[Theorem~5.1]{DamelinNathe2024}; in particular, diametral pairs form only part of that class.

The proof starts from a multilinear majorization of
\[
    (a_1^2+\cdots+a_d^2)^{\alpha/2}
\]
by its values at the corners of \(\prod_i[0,L_i]\). A Bernstein representation converts this majorant into an integral of positive-definite product kernels. The resulting Fourier equality conditions force the support onto a translate of \(H_2\). At the endpoint \(\alpha=1\), the equality set of the convexity step is larger, and a separate rigidity argument is needed to eliminate the additional one-dimensional directions.

Our second exact model is the regular hexagonal torus. Let
\[
    \Lambda_{\mathrm{hex}}
    =
    \mathbb Z b_1+\mathbb Z b_2,
    \qquad
    |b_1|=|b_2|=1,
    \qquad
    b_1\cdot b_2=\frac12,
\]
and set
\[
    T_{\mathrm{hex}}
    =
    \mathbb R^2/\Lambda_{\mathrm{hex}}.
\]
Let
\[
    q=\frac{b_1+b_2}{3},
    \qquad
    H_3=\{0,q,2q\},
\]
and let \(\nu_{H_3}\) be normalized counting measure on \(H_3\).

\begin{theorem}[Regular hexagonal torus]
\label{thm:intro-hexagonal}
For every \(\alpha\geq1\),
\[
    \max_{\mu\in\mathcal P(T_{\mathrm{hex}})}
    I_\alpha(\mu)
    =
    \frac23\,3^{-\alpha/2},
\]
and the maximizers are precisely the translates of \(\nu_{H_3}\).
\end{theorem}

There is no transition in the hexagonal case for \(\alpha\geq1\). The proof is not based on product structure. After reducing to a symmetry chamber, we construct a positive-definite function \(P\) below a geometric majorant \(h\) for the normalized distance, with contact set
\[
    \{P=h\}=H_3.
\]
The spectral lower bound for \(P\), together with the sharp contact set, forces the extremal measure to be uniform on a translate of \(H_3\).

The Fourier and negative-type results above hold for every \(\alpha>0\). The exact global classifications in this paper are restricted to \(\alpha\geq1\). The range \(0<\alpha<1\) is left open here; the convexity and spectral-minorant arguments used for the exact classifications above do not by themselves give a sharp global description in that regime.

Extremal problems for positive powers of distance go back at least to Björck \cite{Bjorck1956}. For the linear distance kernel, Wolf studied the maximal-distance functional and its relation to invariant measures and rendezvous numbers \cite{Wolf1997}; related potential-theoretic and quasihypermetric approaches include \cite{MorrisNickolas1983,FarkasRevesz2006Arch,NickolasWolf2011}. On spheres and compact homogeneous spaces, distance energies are closely connected with the Stolarsky invariance principle and later extremal problems for geodesic and Riesz-type kernels \cite{Stolarsky1973,BilykDaiMatzke2018,BilykMatzkeNathe2024}.

For flat tori, translation invariance makes Fourier analysis natural while the intrinsic distance retains the singular geometry of the cut locus. Fourier and periodic-energy methods on compact manifolds and tori appear in \cite{Randol2014,NechayevaRandol2019,HardinSaffSimanek2014,DamelinNathe2024}; recent torus results using Fourier-positive or linear-programming methods include \cite{Almog2026,Nagel2025,BilykNagelRuohoniemi2026}.

Negative type originates in Schoenberg's work on metric embeddings and positive definite functions \cite{Schoenberg1935,Schoenberg1937,Schoenberg1938}. Generalized roundness was introduced by Enflo and identified with supremal negative type by Lennard, Tonge, and Weston \cite{Enflo1970,LennardTongeWeston1997}; the same exponent appears as the fractional index in manifold-indexed fractional fields \cite{Istas2012}.

The paper is organized as follows. In \cref{sec:preliminaries} we fix the Voronoi and Fourier conventions and record the energy identities used throughout. In \cref{sec:cut-locus} we compute the distributional Hessian of radial distance kernels and extract positive Fourier coefficients along facet normals. In \cref{sec:negative-type} these coefficients are converted into finite negative-type witnesses and destabilizing perturbations of Haar measure. The exact global problems are treated in \cref{sec:rectangular-tori,sec:hexagonal-torus}: the rectangular case uses corner majorization and positive-definite product kernels, while the hexagonal case uses a sharp positive spectral minorant. We conclude with several questions suggested by the Fourier asymptotics and the exact extremizers.

\section{Flat Tori and Distance Energies}
\label{sec:preliminaries}

This section fixes the geometric and Fourier-analytic conventions used throughout the paper. We also record two spectral identities that will be used repeatedly: the Fourier representation of translation-invariant energies for \(L^2\) densities, and the exact second variation at Haar measure.

\subsection{Voronoi geometry and Fourier conventions}
\label{subsec:voronoi-fourier}

Let \(\Lambda\subset\mathbb R^d\) be a full-rank lattice and let
\[
    T_\Lambda=\mathbb R^d/\Lambda
\]
be the associated flat torus. We write \(\pi:\mathbb R^d\to T_\Lambda\) for the quotient map and, when no confusion can arise, use the same notation for a point of \(\mathbb R^d\) and its image in \(T_\Lambda\). The quotient metric is
\[
    \rho_\Lambda(x,y)
    =
    \min_{\lambda\in\Lambda}|x-y-\lambda|.
\]
The minimum is attained because \(\Lambda\) is discrete.

Let
\[
    V=V(\Lambda)
    =
    \bigl\{
        x\in\mathbb R^d:
        |x|\leq |x-\lambda|
        \text{ for every }\lambda\in\Lambda
    \bigr\}
\]
be the Voronoi cell of the origin. It is a centrally symmetric convex polytope whose translates by \(\Lambda\) tile \(\mathbb R^d\) with disjoint interiors. Every point of \(T_\Lambda\) therefore has a representative in \(V\), unique away from \(\pi(\partial V)\), and
\[
    \rho_\Lambda(x,0)=|x|,
    \qquad x\in V.
\]

We denote the Euclidean volume of \(V\) by \(|V|\). If \(m_\Lambda\) is normalized Haar measure on \(T_\Lambda\), then for every integrable \(\Lambda\)-periodic function \(F\),
\begin{equation}
    \int_{T_\Lambda}F(x)\,dm_\Lambda(x)
    =
    \frac{1}{|V|}
    \int_V F(x)\,dx.
    \label{eq:haar-voronoi-normalization}
\end{equation}
The normalization in \eqref{eq:haar-voronoi-normalization} will be important when the singular part of the distributional Hessian is computed in \cref{sec:cut-locus}. Standard background on lattices and Voronoi cells may be found in \cite{ConwaySloane1999}.

A nonzero vector \(v\in\Lambda\) is called \emph{Voronoi relevant} if
\[
    F_v
    =
    \left\{
        x\in V:
        x\cdot v=\frac{|v|^2}{2}
    \right\}
\]
is a \((d-1)\)-dimensional face of \(V\). The identity
\[
    |x|=|x-v|
    \quad\Longleftrightarrow\quad
    2x\cdot v=|v|^2
\]
shows that a point in the relative interior of \(F_v\) has precisely two nearest lattice points, namely \(0\) and \(v\). The opposite facets \(F_v\) and \(F_{-v}\) are identified in the quotient: translation by \(-v\) maps \(F_v\) onto \(F_{-v}\). The cut locus of the origin may accordingly be identified with \(\pi(\partial V)\); intersections of two or more facets form its lower-dimensional strata.

We write
\[
    D_\Lambda
    =
    \operatorname{diam}(T_\Lambda)
    =
    \max_{x\in V}|x|.
\]

The dual lattice is
\[
    \Lambda^*
    =
    \left\{
        \xi\in\mathbb R^d:
        \xi\cdot\lambda\in\mathbb Z
        \text{ for every }\lambda\in\Lambda
    \right\}.
\]
For \(\xi\in\Lambda^*\), set
\[
    e_\xi(x)=e^{2\pi i\xi\cdot x}.
\]
Then \(\{e_\xi:\xi\in\Lambda^*\}\) is an orthonormal basis of \(L^2(T_\Lambda,m_\Lambda)\). For \(f\in L^1(m_\Lambda)\), we use the Fourier convention
\begin{equation}
    \widehat f(\xi)
    =
    \int_{T_\Lambda}
    f(x)e^{-2\pi i\xi\cdot x}\,dm_\Lambda(x),
    \qquad \xi\in\Lambda^*.
    \label{eq:fourier-convention}
\end{equation}
For a finite Borel measure \(\mu\) on \(T_\Lambda\), its Fourier--Stieltjes coefficients are
\[
    \widehat\mu(\xi)
    =
    \int_{T_\Lambda}
    e^{-2\pi i\xi\cdot x}\,d\mu(x).
\]
Thus
\[
    \widehat\mu(0)=\mu(T_\Lambda),
\]
and, for a finite positive measure,
\[
    \widehat\mu(-\xi)
    =
    \overline{\widehat\mu(\xi)}.
\]

We use the convolution convention
\[
    (f*g)(x)
    =
    \int_{T_\Lambda}f(x-y)g(y)\,dm_\Lambda(y),
\]
for which
\[
    \widehat{f*g}(\xi)
    =
    \widehat f(\xi)\widehat g(\xi).
\]

\subsection{Energy functionals and Haar measure}
\label{subsec:energy-haar}

For \(\alpha>0\), define
\[
    K_\alpha(x)
    =
    \rho_\Lambda(x,0)^\alpha
\]
and, for \(\mu\in\mathcal P(T_\Lambda)\),
\begin{equation}
    I_\alpha(\mu)
    =
    \iint_{T_\Lambda\times T_\Lambda}
    \rho_\Lambda(x,y)^\alpha\,
    d\mu(x)\,d\mu(y).
    \label{eq:distance-energy}
\end{equation}
Here \(\mathcal P(T_\Lambda)\) denotes the space of Borel probability measures on \(T_\Lambda\).

More generally, for a real even kernel \(k\in L^1(m_\Lambda)\), write
\[
    \mathcal E_k(\mu)
    =
    \iint_{T_\Lambda\times T_\Lambda}
    k(x-y)\,d\mu(x)\,d\mu(y)
\]
whenever the integral is defined. Thus
\[
    I_\alpha(\mu)=\mathcal E_{K_\alpha}(\mu).
\]

Compactness of the torus immediately gives existence of maximizers.

\begin{proposition}
\label{prop:existence-maximizer}
For every \(\alpha>0\), there exists \(\mu_\alpha\in\mathcal P(T_\Lambda)\) such that
\[
    I_\alpha(\mu_\alpha)
    =
    \max_{\mu\in\mathcal P(T_\Lambda)} I_\alpha(\mu).
\]
\end{proposition}

\begin{proof}
The space \(\mathcal P(T_\Lambda)\) is compact in the weak topology. If \(\mu_n\rightharpoonup\mu\), then \(\mu_n\otimes\mu_n\rightharpoonup\mu\otimes\mu\) on \(T_\Lambda\times T_\Lambda\). Since \((x,y)\mapsto\rho_\Lambda(x,y)^\alpha\) is continuous and bounded,
\[
    I_\alpha(\mu_n)\longrightarrow I_\alpha(\mu).
\]
Hence \(I_\alpha\) attains its maximum.
\end{proof}

The energy is translation invariant. If \(\tau_a(x)=x+a\), then
\[
    I_\alpha((\tau_a)_\#\mu)=I_\alpha(\mu).
\]
Thus uniqueness of a maximizer can only be expected up to translation. Equivalently, if \(X\) and \(Y\) are independent random variables with law \(\mu\), then
\[
    I_\alpha(\mu)
    =
    \mathbb E\bigl[\rho_\Lambda(X,Y)^\alpha\bigr].
\]
This probabilistic form will be convenient in the rectangular case.

We first diagonalize translation-invariant energies at the level of \(L^2\) densities.

\begin{proposition}[Energy diagonalization for \(L^2\) densities]
\label{prop:l2-energy-diagonalization}
Let \(k\in L^1(m_\Lambda)\) be real and even, and let \(r\in L^2(m_\Lambda)\) be real. Then
\begin{equation}
\begin{aligned}
    &\iint_{T_\Lambda\times T_\Lambda}
    k(x-y)r(x)r(y)\,
    dm_\Lambda(x)\,dm_\Lambda(y) \\
    &\hspace{35mm}
    =
    \sum_{\xi\in\Lambda^*}
    \widehat k(\xi)
    |\widehat r(\xi)|^2,
\end{aligned}
    \label{eq:l2-energy-diagonalization}
\end{equation}
and the series on the right converges absolutely.
\end{proposition}

\begin{proof}
Set
\[
    \widetilde r(x)=r(-x),
    \qquad
    h=r*\widetilde r.
\]
Since \(r\in L^2(m_\Lambda)\), Parseval gives
\[
    \sum_{\xi\in\Lambda^*}
    |\widehat r(\xi)|^2
    =
    \|r\|_{L^2(m_\Lambda)}^2<\infty.
\]
Moreover,
\[
    \widehat h(\xi)
    =
    \widehat r(\xi)
    \widehat{\widetilde r}(\xi)
    =
    |\widehat r(\xi)|^2.
\]
The convolution \(h\) is continuous. The series
\[
    \sum_{\xi\in\Lambda^*}
    |\widehat r(\xi)|^2
    e^{2\pi i\xi\cdot z}
\]
converges absolutely and uniformly to a continuous function with the same Fourier coefficients as \(h\). By uniqueness of Fourier coefficients for \(L^1(T_\Lambda)\) functions,
\[
    h(z)
    =
    \sum_{\xi\in\Lambda^*}
    |\widehat r(\xi)|^2
    e^{2\pi i\xi\cdot z}.
\]

By Fubini's theorem,
\[
\begin{aligned}
    &\iint
    k(x-y)r(x)r(y)\,
    dm_\Lambda(x)\,dm_\Lambda(y) \\
    &\hspace{25mm}
    =
    \int_{T_\Lambda}k(z)h(z)\,dm_\Lambda(z).
\end{aligned}
\]
The double integral is absolutely convergent, since Cauchy--Schwarz gives
\[
    \|r*\widetilde r\|_{L^\infty}
    \leq
    \|r\|_{L^2}^2,
\]
and hence
\[
    \int |k(z)|\,|h(z)|\,dm_\Lambda(z)
    \leq
    \|k\|_{L^1}\|r\|_{L^2}^2.
\]
We may therefore integrate the Fourier series of \(h\) term by term. Using the reality and evenness of \(k\),
\[
    \widehat k(-\xi)=\widehat k(\xi),
\]
which yields \eqref{eq:l2-energy-diagonalization}. Finally,
\[
    \sum_{\xi\in\Lambda^*}
    |\widehat k(\xi)|
    |\widehat r(\xi)|^2
    \leq
    \|k\|_{L^1}
    \|r\|_{L^2}^2,
\]
so the series converges absolutely.
\end{proof}

For the exact extremal arguments later in the paper, we shall also need a measure-level identity for kernels with absolutely convergent Fourier series.

\begin{proposition}[Absolutely convergent kernels]
\label{prop:measure-energy-diagonalization}
Let \(k\) be a real even kernel with
\[
    \sum_{\xi\in\Lambda^*}
    |\widehat k(\xi)|<\infty
\]
and
\[
    k(x)
    =
    \sum_{\xi\in\Lambda^*}
    \widehat k(\xi)e^{2\pi i\xi\cdot x}.
\]
Then, for every finite positive Borel measure \(\mu\) on \(T_\Lambda\),
\begin{equation}
    \mathcal E_k(\mu)
    =
    \sum_{\xi\in\Lambda^*}
    \widehat k(\xi)
    |\widehat\mu(\xi)|^2,
    \label{eq:measure-energy-diagonalization}
\end{equation}
and the series converges absolutely.
\end{proposition}

\begin{proof}
The Fourier series of \(k\) converges uniformly, so it may be integrated term by term against \(\mu\otimes\mu\). Hence
\[
\begin{aligned}
    \mathcal E_k(\mu)
    &=
    \sum_{\xi\in\Lambda^*}
    \widehat k(\xi)
    \left(
        \int e^{2\pi i\xi\cdot x}\,d\mu(x)
    \right)
    \left(
        \int e^{-2\pi i\xi\cdot y}\,d\mu(y)
    \right) \\
    &=
    \sum_{\xi\in\Lambda^*}
    \widehat k(\xi)
    |\widehat\mu(\xi)|^2.
\end{aligned}
\]
Since
\[
    |\widehat\mu(\xi)|
    \leq
    \mu(T_\Lambda),
\]
the series is absolutely convergent.
\end{proof}

The preceding identity gives the familiar positive-definite energy bound.

\begin{corollary}
\label{cor:positive-definite-energy}
Under the hypotheses of \cref{prop:measure-energy-diagonalization}, suppose in addition that
\[
    \widehat k(\xi)\geq0
    \qquad
    \text{for every }\xi\in\Lambda^*.
\]
Then every \(\mu\in\mathcal P(T_\Lambda)\) satisfies
\[
    \mathcal E_k(\mu)\geq\widehat k(0).
\]
Equality holds if and only if
\[
    \widehat\mu(\xi)=0
\]
for every nonzero \(\xi\) with \(\widehat k(\xi)>0\).
\end{corollary}

\begin{proof}
Since \(\widehat\mu(0)=1\), \eqref{eq:measure-energy-diagonalization} becomes
\[
    \mathcal E_k(\mu)
    =
    \widehat k(0)
    +
    \sum_{\xi\neq0}
    \widehat k(\xi)|\widehat\mu(\xi)|^2.
\]
The conclusion follows term by term.
\end{proof}

The use of positive- and negative-definite kernels in metric geometry goes back to the classical work of Schoenberg and von Neumann--Schoenberg \cite{Schoenberg1938,vonNeumannSchoenberg1941}; see also \cite{BergChristensenRessel1984} for a systematic treatment, \cite{GuellaMenegatto2017} for strict positive definiteness on tori, and \cite{DamelinHickernellRagozinZeng2010} for the energy/discrepancy and invariant-measure viewpoint.

We finally record the variation of a translation-invariant energy about Haar measure. For brevity, write \(m=m_\Lambda\).

\begin{proposition}[Exact second variation at Haar measure]
\label{prop:haar-second-variation}
Let \(k\in L^1(m)\) be real and even, and let \(f\in L^\infty(m)\) be real with
\[
    \int_{T_\Lambda}f\,dm=0.
\]
For
\[
    d\mu_t=(1+tf)\,dm,
    \qquad
    |t|\,\|f\|_{L^\infty}\leq1,
\]
one has
\begin{equation}
    \mathcal E_k(\mu_t)-\mathcal E_k(m)
    =
    t^2
    \sum_{\xi\in\Lambda^*\setminus\{0\}}
    \widehat k(\xi)
    |\widehat f(\xi)|^2.
    \label{eq:haar-second-variation}
\end{equation}
\end{proposition}

\begin{proof}
The density \(1+tf\) belongs to \(L^2(m)\), so \cref{prop:l2-energy-diagonalization} applies. Since \(f\) has mean zero,
\[
    \widehat{(1+tf)}(0)=1,
\]
while
\[
    \widehat{(1+tf)}(\xi)
    =
    t\widehat f(\xi),
    \qquad \xi\neq0.
\]
Consequently,
\[
    \mathcal E_k(\mu_t)
    =
    \widehat k(0)
    +
    t^2
    \sum_{\xi\neq0}
    \widehat k(\xi)|\widehat f(\xi)|^2,
\]
whereas
\[
    \mathcal E_k(m)=\widehat k(0).
\]
Subtracting gives \eqref{eq:haar-second-variation}.
\end{proof}

A single positive Fourier coefficient is therefore enough to destabilize Haar measure.

\begin{corollary}
\label{cor:positive-fourier-mode}
Let \(k\in L^1(m_\Lambda)\) be real and even. If
\[
    \widehat k(\xi)>0
\]
for some nonzero \(\xi\in\Lambda^*\), then Haar measure is not a local maximizer of \(\mathcal E_k\), even among probability measures whose densities are arbitrarily close to \(1\) in \(L^\infty(m_\Lambda)\).

More precisely, for \(0<|t|\leq1\),
\[
    d\mu_t(x)
    =
    \bigl(
        1+t\cos(2\pi\xi\cdot x)
    \bigr)\,dm_\Lambda(x)
\]
is a probability measure and
\[
    \mathcal E_k(\mu_t)-\mathcal E_k(m_\Lambda)
    =
    \frac{t^2}{2}\,\widehat k(\xi)>0.
\]
\end{corollary}

\begin{proof}
The density is nonnegative and has integral one. The only nonzero Fourier coefficients of
\[
    x\longmapsto\cos(2\pi\xi\cdot x)
\]
are \(1/2\) at \(\xi\) and \(-\xi\). The conclusion follows directly from \eqref{eq:haar-second-variation}.
\end{proof}

We shall use the measure-level identity only for kernels with absolutely convergent Fourier series. For the distance kernels studied in \cref{sec:cut-locus}, the relevant Fourier coefficients enter instead through the \(L^2\) variation formula \eqref{eq:haar-second-variation}.

\section{Cut Loci and Fourier Coefficients}
\label{sec:cut-locus}

By \cref{cor:positive-fourier-mode}, to destabilize Haar measure it is enough to find a positive nonzero Fourier coefficient of the distance kernel. We now show that such coefficients are forced by the geometry of the cut locus. More precisely, each Voronoi facet gives rise to a sequence of high frequencies approaching its normal direction along which the Fourier coefficients are positive.

Throughout this section we assume that \(d\geq2\).

\subsection{The distributional Hessian}
\label{subsec:distributional-hessian}

The distributional Hessian of the distance from a point on a general Riemannian manifold, including its singular contribution on the cut locus, was described by Mantegazza, Mascellani, and Uraltsev \cite{MantegazzaMascellaniUraltsev2014}. On a flat torus, the Voronoi structure makes the singular contribution along each facet explicit. The next proposition gives the facet coefficient and Haar normalization used in the Fourier asymptotics below. For related face-sensitive Fourier asymptotics associated with polytopes, see \cite{BrandoliniColzaniGariboldiGiganteMonguzzi2023}.

Let \(D=D_\Lambda\), and let
\[
    \Phi:[0,D]\longrightarrow\mathbb R.
\]
We consider the radial distance kernel
\[
    K_\Phi(x)=\Phi\bigl(\rho_\Lambda(x,0)\bigr).
\]
We impose the following assumptions:
\begin{enumerate}
\item[\textup{(H1)}]
\(\Phi\in C([0,D])\cap C^1((0,D])\), and \(\Phi'\) is locally absolutely continuous on \((0,D]\);

\item[\textup{(H2)}]
for some \(\delta>0\),
    \[
        \int_0^\delta
        \left(
            |\Phi''(r)|+\frac{|\Phi'(r)|}{r}
        \right)r^{d-1}\,dr<\infty;
    \]

\item[\textup{(H3)}]
    \[
        \lim_{r\downarrow0}r^{d-1}\Phi'(r)=0.
    \]
\end{enumerate}

For a unit vector \(e\in S^{d-1}\), write \(D_e=e\cdot\nabla\). Inside the Voronoi cell, away from the origin, \(K_\Phi(x)=\Phi(|x|)\). Thus, with \(r=|x|\),
\begin{equation}
    D_e^2K_\Phi(x)
    =
    \Phi''(r)\left(\frac{e\cdot x}{r}\right)^2
    +
    \frac{\Phi'(r)}{r}
    \left[
        1-\left(\frac{e\cdot x}{r}\right)^2
    \right]
    =:H_e^{\mathrm{reg}}(x).
    \label{eq:regular-directional-hessian}
\end{equation}
By \textup{(H2)}, \(H_e^{\mathrm{reg}}\in L^1(m_\Lambda)\). The same condition makes the classical gradient locally integrable: in radial coordinates,
\[
    \int_0^\delta |\Phi'(r)|r^{d-1}\,dr
    \leq
    \delta\int_0^\delta |\Phi'(r)|r^{d-2}\,dr<\infty.
\]
Since \(K_\Phi\) is continuous across paired Voronoi facets, its first distributional derivative carries no facet measure; the singular contribution appears only after differentiating once more.

Let \(v\in\Lambda\setminus\{0\}\) be Voronoi relevant and let \(F_v\) be the corresponding facet. At a point \(x\in\operatorname{relint}F_v\), the two nearest lattice points are \(0\) and \(v\). Since \(|x|=|x-v|\), the two one-sided gradients are
\[
    \frac{\Phi'(|x|)}{|x|}\,x
    \qquad\text{and}\qquad
    \frac{\Phi'(|x|)}{|x|}\,(x-v).
\]
Their difference is therefore
\[
    \frac{\Phi'(|x|)}{|x|}\,v.
\]
If \(n_v=v/|v|\), the corresponding jump of the normal derivative has magnitude
\[
    \Phi'(|x|)\frac{|v|}{|x|}.
\]

Let \(\mathcal R\) be the set of Voronoi relevant vectors, and choose \(\mathcal R_+\subset\mathcal R\) containing exactly one vector from each pair \(\{w,-w\}\). Whenever \(\mathcal H^{d-1}\!\restriction F_w\) is regarded below as a measure on \(T_\Lambda\), we mean its pushforward \(\pi_\#(\mathcal H^{d-1}\!\restriction F_w)\) under the quotient map \(\pi:\mathbb R^d\to T_\Lambda\); equivalently, facet integrals of torus functions are evaluated using their \(\Lambda\)-periodic lifts.

\begin{lemma}[Distributional Hessian]
\label{lem:distributional-hessian}
Assume \textup{(H1)--(H3)}. For every unit vector \(e\in S^{d-1}\) and every \(\psi\in C^\infty(T_\Lambda)\),
\begin{align}
    \langle D_e^2K_\Phi,\psi\rangle
    ={}&
    \int_{T_\Lambda}
    H_e^{\mathrm{reg}}(x)\psi(x)\,dm_\Lambda(x)
    \notag\\
    &-
    \frac{1}{|V|}
    \sum_{w\in\mathcal R_+}
    \int_{F_w}
    \frac{\Phi'(|x|)}{|x|}
    \frac{(e\cdot w)^2}{|w|}
    \psi(x)\,dS(x).
    \label{eq:distributional-hessian}
\end{align}
Equivalently,
\[
    D_e^2K_\Phi
    =
    H_e^{\mathrm{reg}}\,m_\Lambda
    -
    \frac{1}{|V|}
    \sum_{w\in\mathcal R_+}
    \frac{\Phi'(|x|)}{|x|}
    \frac{(e\cdot w)^2}{|w|}
    \,\mathcal H^{d-1}\!\restriction F_w
\]
as a Radon measure on \(T_\Lambda\).
\end{lemma}

\begin{proof}
Identify functions on the torus with their \(\Lambda\)-periodic lifts. Choose \(\varepsilon>0\) so that \(\overline{B(0,\varepsilon)}\subset V^\circ\), and set
\[
    V_\varepsilon
    =
    V^\circ\setminus\overline{B(0,\varepsilon)}.
\]

We first identify the distributional first derivative. Integration by parts on \(V_\varepsilon\) produces boundary terms on the facets of \(V\) and on \(\partial B(0,\varepsilon)\). The facet terms cancel in opposite pairs. Indeed, if \(x\in F_w\), then \(2x\cdot w=|w|^2\), so \(|x|=|x-w|\) and hence \(K_\Phi(x)=K_\Phi(x-w)\); moreover the test function has the same values at \(x\) and \(x-w\) by \(\Lambda\)-periodicity, while the corresponding outer normals are opposite. The contribution from the small sphere is
\[
    O\!\left(
        \varepsilon^{d-1}|\Phi(\varepsilon)|
    \right),
\]
and hence tends to zero. It follows that
\[
    D_eK_\Phi(x)
    =
    \Phi'(|x|)\frac{e\cdot x}{|x|}
\]
in the distributional sense.

We differentiate once more. Using \eqref{eq:haar-voronoi-normalization} and integrating by parts on \(V_\varepsilon\),
\[
\begin{aligned}
    -\frac{1}{|V|}
    \int_{V_\varepsilon}
    D_eK_\Phi\,D_e\psi\,dx
    ={}&
    \frac{1}{|V|}
    \int_{V_\varepsilon}
    H_e^{\mathrm{reg}}\psi\,dx \\
    &-
    \frac{1}{|V|}
    \int_{\partial V_\varepsilon}
    D_eK_\Phi\,(e\cdot n)\psi\,dS.
\end{aligned}
\]
The contribution of \(\partial B(0,\varepsilon)\) is
\[
    O\!\left(
        \varepsilon^{d-1}|\Phi'(\varepsilon)|
    \right),
\]
which tends to zero by \textup{(H3)}.

It remains to pair opposite facets. Fix \(w\in\mathcal R_+\). For \(x\in F_w\), the point \(x-w\) lies in \(F_{-w}\) and represents the same point of the torus. Hence \(\psi(x-w)=\psi(x)\). The outer normals to \(F_w\) and \(F_{-w}\) are \(w/|w|\) and \(-w/|w|\), respectively, while
\[
    D_eK_\Phi(x)-D_eK_\Phi(x-w)
    =
    \frac{\Phi'(|x|)}{|x|}\,e\cdot w.
\]
After translating the integral over \(F_{-w}\) to \(F_w\), the two boundary contributions combine to
\[
    \int_{F_w}
    \frac{\Phi'(|x|)}{|x|}
    \frac{(e\cdot w)^2}{|w|}
    \psi(x)\,dS(x).
\]
The boundary term enters with a minus sign, giving \eqref{eq:distributional-hessian}.

Intersections of two or more facets have \(\mathcal H^{d-1}\)-measure zero, so no additional term occurs on the lower-dimensional strata. Letting \(\varepsilon\downarrow0\) completes the proof.
\end{proof}

\subsection{Frequencies aligned with Voronoi facets}
\label{subsec:aligned-frequencies}

Fix a Voronoi relevant vector \(v\). We shall use frequencies in the even dual lattice \(2\Lambda^*\) whose directions approach the normal direction of \(F_v\). The factor \(2\) locks the phase on the affine hyperplane containing the facet.

\begin{lemma}[Aligned even dual frequencies]
\label{lem:aligned-even-frequencies}
For every nonzero \(v\in\mathbb R^d\), there exist \(t_j\to\infty\) and \(\xi_j\in2\Lambda^*\) such that
\[
    |\xi_j-t_jv|\longrightarrow0.
\]
Consequently,
\[
    |\xi_j|\longrightarrow\infty,
    \qquad
    \operatorname{dist}(\xi_j,\mathbb Rv)\longrightarrow0,
    \qquad
    \frac{\xi_j}{|\xi_j|}
    \longrightarrow
    \frac{v}{|v|}.
\]
\end{lemma}

\begin{proof}
Choose a basis \(\gamma_1,\ldots,\gamma_d\) of \(2\Lambda^*\) and write
\[
    v=\sum_{i=1}^d\theta_i\gamma_i.
\]
By simultaneous Dirichlet approximation, for every integer \(Q\geq1\) there exist
\[
    1\leq q\leq Q^d,
    \qquad
    p_1,\ldots,p_d\in\mathbb Z,
\]
such that
\[
    \max_{1\leq i\leq d}|q\theta_i-p_i|
    \leq\frac1Q.
\]
Set
\[
    t=q,
    \qquad
    \xi=\sum_{i=1}^d p_i\gamma_i\in2\Lambda^*.
\]
Since all norms on a finite-dimensional vector space are equivalent,
\[
    |\xi-tv|
    \leq
    \frac{C_\Lambda}{Q}
\]
for a constant \(C_\Lambda\) independent of \(Q\).

If the resulting integers \(q\) are unbounded along a sequence \(Q\to\infty\), an appropriate subsequence gives the assertion. Otherwise, after passing to a subsequence, some fixed \(q\) occurs for arbitrarily large \(Q\). It follows that
\[
    q\theta_i\in\mathbb Z
    \qquad
    \text{for every }i,
\]
and hence \(qv\in2\Lambda^*\). We may then take
\[
    t_j=jq,
    \qquad
    \xi_j=jqv.
\]
The remaining assertions follow immediately.
\end{proof}

For \(v\in\Lambda\) and \(\xi\in2\Lambda^*\), we also have
\[
    \frac{\xi\cdot v}{2}\in\mathbb Z.
\]
Thus the normal component of the Fourier phase is constant on \(F_v\) modulo integers.

\subsection{Positive high-frequency coefficients}
\label{subsec:positive-fourier-coefficients}

We now combine the distributional formula with the aligned frequencies. The regular part disappears at high frequency, while the contribution of the selected facet survives.

\begin{theorem}[Positive Fourier coefficients along facet normals]
\label{thm:positive-fourier-tail}
Assume \textup{(H1)--(H3)} and suppose that
\[
    \Phi'(r)>0,
    \qquad 0<r\leq D_\Lambda.
\]
Let \(v\) be Voronoi relevant. Then there exist \(t_j\to\infty\) and \(\xi_j\in2\Lambda^*\) such that
\[
    |\xi_j-t_jv|\longrightarrow0
\]
and
\begin{equation}
    (2\pi|\xi_j|)^2\widehat K_\Phi(\xi_j)
    \longrightarrow
    \frac{|v|}{|V|}
    \int_{F_v}
    \frac{\Phi'(|x|)}{|x|}\,dS(x)
    >0.
    \label{eq:positive-fourier-tail}
\end{equation}
In particular,
\[
    \widehat K_\Phi(\xi_j)>0
\]
for all sufficiently large \(j\).
\end{theorem}

\begin{proof}
Choose \(\mathcal R_+\) so that \(v\in\mathcal R_+\), and let \(\xi_j,t_j\) be as in \cref{lem:aligned-even-frequencies}. Set
\[
    e_j=\frac{\xi_j}{|\xi_j|},
    \qquad
    e=\frac{v}{|v|}.
\]
Then \(e_j\to e\).

Taking the \(\xi_j\)-th Fourier coefficient of \eqref{eq:distributional-hessian}, with direction \(e_j\), gives
\begin{align}
    -(2\pi|\xi_j|)^2\widehat K_\Phi(\xi_j)
    ={}&
    \widehat{H_{e_j}^{\mathrm{reg}}}(\xi_j)
    \notag\\
    &-
    \frac{1}{|V|}
    \sum_{w\in\mathcal R_+}
    \int_{F_w}
    \frac{\Phi'(|x|)}{|x|}
    \frac{(e_j\cdot w)^2}{|w|}
    e^{-2\pi i\xi_j\cdot x}\,dS(x).
    \label{eq:fourier-transformed-hessian}
\end{align}

We first consider the regular part. For every unit vector \(u\), \eqref{eq:regular-directional-hessian} gives the pointwise bound
\[
    |H_u^{\mathrm{reg}}(x)|
    \leq
    |\Phi''(|x|)|
    +
    \frac{|\Phi'(|x|)|}{|x|}
    \qquad\text{for a.e. }x\in V.
\]
The right-hand side is integrable near the origin by \textup{(H2)}, and is integrable on compact subsets of \(V\setminus\{0\}\) by \textup{(H1)}. Since \(e_j\to e\), we also have \(H_{e_j}^{\mathrm{reg}}(x)\to H_e^{\mathrm{reg}}(x)\) for almost every \(x\). Dominated convergence therefore yields
\[
    \|H_{e_j}^{\mathrm{reg}}-H_e^{\mathrm{reg}}\|_{L^1(m_\Lambda)}
    \longrightarrow0.
\]
Hence
\[
\begin{aligned}
    \left|
        \widehat{H_{e_j}^{\mathrm{reg}}}(\xi_j)
    \right|
    \leq{}&
    \|H_{e_j}^{\mathrm{reg}}-H_e^{\mathrm{reg}}\|_{L^1}
    +
    \left|
        \widehat{H_e^{\mathrm{reg}}}(\xi_j)
    \right|
    \longrightarrow0
\end{aligned}
\]
by the Riemann--Lebesgue lemma.

Now consider \(F_v\). Write
\[
    x=\frac{v}{2}+y,
    \qquad
    y\cdot v=0.
\]
Since \(\xi_j\in2\Lambda^*\),
\[
    e^{-2\pi i\xi_j\cdot v/2}=1.
\]
Writing
\[
    \xi_j=t_jv+\delta_j,
    \qquad
    \delta_j\longrightarrow0,
\]
we have
\[
    \xi_j\cdot y=\delta_j\cdot y\longrightarrow0
\]
uniformly on \(F_v\). Therefore
\[
    e^{-2\pi i\xi_j\cdot x}\longrightarrow1
\]
uniformly on \(F_v\), while
\[
    \frac{(e_j\cdot v)^2}{|v|}
    \longrightarrow |v|.
\]
It follows that the contribution of \(F_v\) tends to
\[
    |v|
    \int_{F_v}
    \frac{\Phi'(|x|)}{|x|}\,dS(x).
\]

It remains to consider \(F_w\) with \(w\in\mathcal R_+\) and \(w\neq v\). Distinct vectors in \(\mathcal R_+\) are not parallel. Indeed, if \(u\) and \(cu\), with \(c>1\), are nonzero lattice vectors, then every point on the bisector for \(cu\) satisfies
\[
    x\cdot u=\frac{c|u|^2}{2}>\frac{|u|^2}{2},
\]
so \(cu\) cannot be Voronoi relevant. Hence two relevant vectors on the same line differ only by sign. If \(P_w\) denotes orthogonal projection onto \(w^\perp\), then
\[
    P_wv\neq0,
\]
and therefore
\[
    |P_w\xi_j|
    =
    t_j|P_wv|+o(t_j)
    \longrightarrow\infty.
\]
Set, on \(F_w\),
\[
    A_j(x)
    =
    \frac{\Phi'(|x|)}{|x|}
    \frac{(e_j\cdot w)^2}{|w|},
    \qquad
    A(x)
    =
    \frac{\Phi'(|x|)}{|x|}
    \frac{(e\cdot w)^2}{|w|}.
\]
Because \(F_w\) is compact and separated from the origin and \(e_j\to e\), we have
\[
    \|A_j-A\|_{L^1(F_w)}\longrightarrow0.
\]
Write \(x=w/2+y\) on \(F_w\), with \(y\in w^\perp\). Then
\[
    e^{-2\pi i\xi_j\cdot x}
    =
    e^{-\pi i\xi_j\cdot w}
    e^{-2\pi i(P_w\xi_j)\cdot y}.
\]
The first factor is constant on the facet and has modulus one (indeed it equals \(1\), since \(\xi_j\in2\Lambda^*\) and \(w\in\Lambda\)); all oscillation is therefore tangential. Since \( |P_w\xi_j|\to\infty\), the Riemann--Lebesgue lemma on \(w^\perp\), applied to the fixed amplitude \(A\mathbf 1_{F_w}\in L^1(w^\perp)\), gives
\[
    \int_{F_w}A(x)e^{-2\pi i\xi_j\cdot x}\,dS(x)
    \longrightarrow0.
\]
Therefore
\[
\begin{aligned}
    \left|
        \int_{F_w}A_j(x)e^{-2\pi i\xi_j\cdot x}\,dS(x)
    \right|
    &\leq
    \|A_j-A\|_{L^1(F_w)}
    +
    \left|
        \int_{F_w}A(x)e^{-2\pi i\xi_j\cdot x}\,dS(x)
    \right| \\
    &\longrightarrow0.
\end{aligned}
\]

Passing to the limit in \eqref{eq:fourier-transformed-hessian} yields \eqref{eq:positive-fourier-tail}. The limiting integral is strictly positive because \(F_v\) has positive \((d-1)\)-dimensional measure, is separated from the origin, and \(\Phi'>0\).
\end{proof}

For distance powers, the theorem specializes as follows.

\begin{corollary}[Distance powers]
\label{cor:positive-fourier-distance-powers}
Let \(\alpha>0\). For every Voronoi relevant vector \(v\), there exist \(\xi_j\in2\Lambda^*\) such that
\[
    |\xi_j|\longrightarrow\infty,
    \qquad
    \operatorname{dist}(\xi_j,\mathbb Rv)\longrightarrow0,
\]
and
\begin{equation}
    (2\pi|\xi_j|)^2
    \widehat{\rho_\Lambda^\alpha}(\xi_j)
    \longrightarrow
    \frac{\alpha|v|}{|V|}
    \int_{F_v}|x|^{\alpha-2}\,dS(x)
    >0.
    \label{eq:positive-fourier-distance-powers}
\end{equation}
\end{corollary}

\begin{proof}
Take \(\Phi(r)=r^\alpha\). Near the origin,
\[
    \left(
        |\Phi''(r)|+\frac{|\Phi'(r)|}{r}
    \right)r^{d-1}
    \lesssim
    r^{\alpha+d-3},
\]
which is integrable because \(d\geq2\) and \(\alpha>0\). Moreover,
\[
    r^{d-1}\Phi'(r)
    =
    \alpha r^{\alpha+d-2}
    \longrightarrow0.
\]
Thus \textup{(H1)--(H3)} hold, and
\[
    \frac{\Phi'(r)}{r}
    =
    \alpha r^{\alpha-2}.
\]
The result follows from \cref{thm:positive-fourier-tail}.
\end{proof}

\section{Negative Type and Haar Instability}
\label{sec:negative-type}

The positive Fourier coefficients obtained in \cref{cor:positive-fourier-distance-powers} have two consequences. They produce finite witnesses to the failure of negative type, and they yield arbitrarily small perturbations of Haar measure that increase the distance energy.

\subsection{Negative type and generalized roundness}
\label{subsec:negative-type-roundness}

Let \((X,\rho)\) be a metric space and let \(p>0\). We say that \(X\) has \emph{\(p\)-negative type} if, for every finite family \(x_1,\ldots,x_N\in X\) and every choice of real numbers \(c_1,\ldots,c_N\) satisfying
\[
    \sum_{i=1}^N c_i=0,
\]
one has
\begin{equation}
    \sum_{i,j=1}^N
    c_ic_j\,\rho(x_i,x_j)^p
    \leq0.
    \label{eq:p-negative-type}
\end{equation}
The \emph{supremal negative type} of \(X\) is
\[
    \wp(X)
    =
    \sup\Bigl(
        \{p>0:X\text{ has \(p\)-negative type}\}\cup\{0\}
    \Bigr).
\]
In the probability literature this same quantity is called the \emph{fractional index} \(\beta_X\); see Istas \cite[\S2.6]{Istas2012}. Thus \(\beta_X=\wp(X)\) with the present conventions.

This notion goes back to Schoenberg's work on metric embeddings and positive definite functions \cite{Schoenberg1935,Schoenberg1938}. It is equivalent to generalized roundness in the following sense. A metric space \(X\) has \emph{generalized roundness \(q\)} if, for every \(n\geq2\) and every choice of
\[
    a_1,\ldots,a_n,b_1,\ldots,b_n\in X,
\]
one has
\[
\begin{aligned}
    &\sum_{1\leq i<j\leq n}
    \bigl(
        \rho(a_i,a_j)^q+\rho(b_i,b_j)^q
    \bigr) \\
    &\hspace{35mm}
    \leq
    \sum_{i,j=1}^n
    \rho(a_i,b_j)^q.
\end{aligned}
\]
Its generalized roundness is
\[
    \operatorname{gr}(X)
    =
    \sup\Bigl(
        \{q>0:X\text{ has generalized roundness \(q\)}\}
        \cup\{0\}
    \Bigr).
\]
By the theorem of Lennard, Tonge, and Weston \cite{LennardTongeWeston1997},
\begin{equation}
    \operatorname{gr}(X)=\wp(X).
    \label{eq:roundness-negative-type}
\end{equation}

For a finite signed Borel measure \(\nu\) on a compact metric space \((X,\rho)\), define
\[
    \mathcal Q_p(\nu)
    =
    \iint_{X\times X}
    \rho(x,y)^p\,d\nu(x)\,d\nu(y).
\]
If
\[
    \nu=\sum_{i=1}^N c_i\delta_{x_i},
    \qquad
    \nu(X)=0,
\]
then
\[
    \mathcal Q_p(\nu)
    =
    \sum_{i,j=1}^N
    c_ic_j\,\rho(x_i,x_j)^p.
\]
Thus a finitely supported zero-mass measure with \(\mathcal Q_p(\nu)>0\) witnesses the failure of \(p\)-negative type. This signed-measure formulation is also natural in the distance geometry of compact metric spaces; see, for example, \cite{NickolasWolf2009,NickolasWolf2011,NickolasWolf2011III}.

\subsection{From positive Fourier modes to finite witnesses}
\label{subsec:finite-negative-type-witnesses}

We first pass directly from a positive Fourier coefficient to a signed zero-mass measure.

\begin{proposition}[Fourier witness]
\label{prop:fourier-negative-type-witness}
Let \(p>0\), and suppose that
\[
    \widehat{\rho_\Lambda^p}(\xi)>0
\]
for some nonzero \(\xi\in\Lambda^*\). Define
\[
    d\nu_\xi(x)
    =
    \cos(2\pi\xi\cdot x)\,dm_\Lambda(x).
\]
Then
\[
    \nu_\xi(T_\Lambda)=0
\]
and
\begin{equation}
    \mathcal Q_p(\nu_\xi)
    =
    \frac12\,\widehat{\rho_\Lambda^p}(\xi)
    >0.
    \label{eq:continuous-negative-type-witness}
\end{equation}
\end{proposition}

\begin{proof}
Since \(\xi\neq0\), orthogonality of characters gives
\[
    \int_{T_\Lambda}
    \cos(2\pi\xi\cdot x)\,dm_\Lambda(x)=0.
\]
The only nonzero Fourier coefficients of
\[
    f_\xi(x)=\cos(2\pi\xi\cdot x)
\]
are
\[
    \widehat f_\xi(\xi)
    =
    \widehat f_\xi(-\xi)
    =
    \frac12.
\]
Applying \cref{prop:l2-energy-diagonalization} with \(k=\rho_\Lambda^p\) and \(r=f_\xi\), we obtain
\[
\begin{aligned}
    \mathcal Q_p(\nu_\xi)
    &=
    \frac14\widehat{\rho_\Lambda^p}(\xi)
    +
    \frac14\widehat{\rho_\Lambda^p}(-\xi) \\
    &=
    \frac12\widehat{\rho_\Lambda^p}(\xi),
\end{aligned}
\]
since the distance kernel is real and even.
\end{proof}

To return to the finite formulation \eqref{eq:p-negative-type}, we use the following approximation.

\begin{lemma}[Atomic approximation]
\label{lem:atomic-negative-type-approximation}
Let \((X,\rho)\) be a compact metric space, let \(p>0\), and let \(\nu\) be a finite signed Borel measure on \(X\) such that
\[
    \nu(X)=0.
\]
If
\[
    \mathcal Q_p(\nu)>0,
\]
then there exist points \(x_1,\ldots,x_N\in X\) and real numbers \(c_1,\ldots,c_N\) satisfying
\[
    \sum_{i=1}^N c_i=0
\]
and
\[
    \sum_{i,j=1}^N
    c_ic_j\,\rho(x_i,x_j)^p>0.
\]
\end{lemma}

\begin{proof}
Set
\[
    K(x,y)=\rho(x,y)^p.
\]
Since \(X\times X\) is compact, \(K\) is uniformly continuous. Let
\[
    \omega_K(s)
    =
    \sup\bigl\{
        |K(x,y)-K(x',y')|:
        \rho(x,x')+\rho(y,y')\leq s
    \bigr\}.
\]
Then \(\omega_K(s)\to0\) as \(s\downarrow0\).

For \(\delta>0\), choose a finite Borel partition
\[
    X=Q_1\sqcup\cdots\sqcup Q_N
\]
with \(\operatorname{diam}(Q_i)\leq\delta\), and choose \(x_i\in Q_i\). Define
\[
    c_i=\nu(Q_i),
    \qquad
    \nu_\delta=\sum_{i=1}^N c_i\delta_{x_i}.
\]
Since the \(Q_i\) form a partition,
\[
    \sum_{i=1}^N c_i=\nu(X)=0.
\]

For \(x\in Q_i\) and \(y\in Q_j\),
\[
    |K(x,y)-K(x_i,x_j)|
    \leq
    \omega_K(2\delta).
\]
Consequently,
\[
\begin{aligned}
    \bigl|
        \mathcal Q_p(\nu)
        -
        \mathcal Q_p(\nu_\delta)
    \bigr|
    &\leq
    \omega_K(2\delta)
    \sum_{i,j=1}^N
    |\nu|(Q_i)|\nu|(Q_j) \\
    &=
    \omega_K(2\delta)\,|\nu|(X)^2.
\end{aligned}
\]
Hence
\[
    \mathcal Q_p(\nu_\delta)
    \longrightarrow
    \mathcal Q_p(\nu)>0.
\]
For sufficiently small \(\delta\), the atomic measure \(\nu_\delta\) therefore has positive quadratic form.
\end{proof}

We now obtain the metric consequence of the cut-locus theorem.

\begin{theorem}[Vanishing generalized roundness]
\label{thm:generalized-roundness-zero}
Let \(\Lambda\subset\mathbb R^d\) be a full-rank lattice with \(d\geq2\). Then
\begin{equation}
    \wp(T_\Lambda)=0
    \qquad\text{and}\qquad
    \operatorname{gr}(T_\Lambda)=0.
    \label{eq:generalized-roundness-zero}
\end{equation}
Equivalently, \(T_\Lambda\) fails to have \(p\)-negative type for every \(p>0\).
\end{theorem}

\begin{proof}
Fix \(p>0\), and choose a Voronoi relevant vector \(v\). By \cref{cor:positive-fourier-distance-powers}, there is a nonzero \(\xi\in\Lambda^*\) such that
\[
    \widehat{\rho_\Lambda^p}(\xi)>0.
\]
The measure from \cref{prop:fourier-negative-type-witness} then satisfies
\[
    \nu_\xi(T_\Lambda)=0,
    \qquad
    \mathcal Q_p(\nu_\xi)>0.
\]
Applying \cref{lem:atomic-negative-type-approximation}, we obtain finitely many points \(x_1,\ldots,x_N\in T_\Lambda\) and real coefficients \(c_1,\ldots,c_N\), with
\[
    \sum_{i=1}^N c_i=0,
\]
such that
\[
    \sum_{i,j=1}^N
    c_ic_j\,\rho_\Lambda(x_i,x_j)^p>0.
\]
Thus \(T_\Lambda\) does not have \(p\)-negative type.

Since \(p>0\) was arbitrary,
\[
    \wp(T_\Lambda)=0.
\]
Equation \eqref{eq:roundness-negative-type} then gives
\[
    \operatorname{gr}(T_\Lambda)=0.
\]
\end{proof}

\begin{remark}[Relation to Venet's cylinder obstruction]
\label{rem:venet-alternative-proof}
The conclusion of \cref{thm:generalized-roundness-zero} also follows from Venet's theorem that every product cylinder has fractional index zero \cite[Theorem~1]{Venet2019}. Choose a shortest nonzero vector \(v\in\Lambda\), and let \(e\perp v\) be a unit vector. For every \(w\in\Lambda\setminus\{0\}\) and \(|t|\leq1/2\),
\[
    2t\,v\cdot w
    \leq
    2|t|\,|v\cdot w|
    \leq
    |v|\,|w|
    \leq
    |w|^2,
\]
so \([-v/2,v/2]\subset V\). At \(v/2\), equality in the Voronoi inequality can occur only for \(w=v\): if \(v\cdot w=|w|^2\), then \(|w|\leq|v|\), and the minimality of \(|v|\) together with equality in Cauchy--Schwarz forces \(w=v\). Thus \(v/2\in\operatorname{relint}F_v\), and similarly \(-v/2\in\operatorname{relint}F_{-v}\).

Let \(\mathcal R\) be the finite set of Voronoi relevant vectors. For every \(w\in\mathcal R\setminus\{v,-v\}\), the corresponding Voronoi inequality is strict on the compact segment \(\{av:|a|\leq1/2\}\). Hence, for some \(\varepsilon>0\),
\[
    av+se\in V
    \qquad
    \text{whenever}
    \qquad
    |a|\leq\frac12,
    \quad
    |s|<2\varepsilon.
\]
It follows that
\[
    (\mathbb R/\mathbb Z)\times(-\varepsilon,\varepsilon)
    \longrightarrow T_\Lambda,
    \qquad
    (t,s)\longmapsto tv+se\pmod\Lambda,
\]
is an isometric embedding when the circle has circumference \(|v|\). Indeed, after reducing the difference of two circle parameters to \(a\in[-1/2,1/2]\), the corresponding vector \(av+(s_1-s_2)e\) lies in \(V\), and its quotient distance is
\[
    \sqrt{a^2|v|^2+(s_1-s_2)^2}.
\]
Negative type passes to metric subspaces, and rescaling does not change supremal negative type. Venet's cylinder theorem therefore gives \(\wp(T_\Lambda)=0\).
\end{remark}

\subsection{Instability of Haar measure}
\label{subsec:haar-instability}

Haar instability now follows immediately from the variation formula in \cref{subsec:energy-haar}.

\begin{corollary}[Haar instability]
\label{cor:haar-instability-distance-powers}
Let \(\Lambda\subset\mathbb R^d\) be a full-rank lattice with \(d\geq2\), and let \(p>0\). For every \(\varepsilon>0\), there exists a smooth probability density \(r\) on \(T_\Lambda\) such that
\[
    \|r-1\|_{L^\infty(m_\Lambda)}<\varepsilon
\]
and
\[
    I_p(r\,m_\Lambda)>I_p(m_\Lambda).
\]
In particular, Haar measure is not a local maximizer of \(I_p\), even among smooth probability densities in the \(L^\infty\) topology.
\end{corollary}

\begin{proof}
By \cref{cor:positive-fourier-distance-powers}, there is a nonzero \(\xi\in\Lambda^*\) with
\[
    \widehat{\rho_\Lambda^p}(\xi)>0.
\]
The conclusion follows from \cref{cor:positive-fourier-mode}.
\end{proof}

\section{Rectangular Tori}
\label{sec:rectangular-tori}

The general argument of the preceding sections does not determine the global maximizers of the distance energy. For orthogonal product tori in dimension \(d\geq2\), however, the full problem can be solved for every \(\alpha\geq1\). The proof combines separate convexity with a family of positive-definite product kernels. The maximization of positive powers of distance has a classical history going back to Björck \cite{Bjorck1956}; here the product geometry leads to an exact classification of the maximizing measures.

Let
\[
    T_{\mathbf L}^d
    =
    \prod_{i=1}^d \mathbb R/(2L_i\mathbb Z),
    \qquad
    L_i>0,
\]
and write
\[
    d_i(x_i,y_i)
    =
    \min_{k\in\mathbb Z}|x_i-y_i-2kL_i|
    \in[0,L_i]
\]
for the distance in the \(i\)-th coordinate. Then
\[
    \rho(x,y)^2
    =
    \sum_{i=1}^d d_i(x_i,y_i)^2.
\]
Set
\[
    D^2=\sum_{i=1}^d L_i^2,
\]
so that \(D=\operatorname{diam}(T_{\mathbf L}^d)\).

We shall use the finite subgroup
\[
    H_2
    =
    \prod_{i=1}^d\{0,L_i\}
    \subset T_{\mathbf L}^d.
\]
For \(a\in T_{\mathbf L}^d\), let
\[
    \mu_a
    =
    2^{-d}
    \sum_{\varepsilon\in\{0,1\}^d}
    \delta_{a+(\varepsilon_1L_1,\ldots,\varepsilon_dL_d)}
\]
be normalized counting measure on the translate \(a+H_2\).

The three regimes are as follows.

\begin{theorem}[Exact maximizers on rectangular tori]
\label{thm:rectangular-phase-diagram}
Let \(d\geq2\).

\begin{enumerate}
\item[\textup{(i)}]
If \(1\leq\alpha<2\), then
\begin{equation}
    \max_{\mu\in\mathcal P(T_{\mathbf L}^d)}
    I_\alpha(\mu)
    =
    2^{-d}
    \sum_{\varepsilon\in\{0,1\}^d}
    \left(
        \sum_{i=1}^d\varepsilon_iL_i^2
    \right)^{\alpha/2}.
    \label{eq:rectangular-subquadratic-maximum}
\end{equation}
The maximizers are precisely the measures \(\mu_a\), \(a\in T_{\mathbf L}^d\).

\item[\textup{(ii)}]
For \(\alpha=2\),
\begin{equation}
    \max_{\mu\in\mathcal P(T_{\mathbf L}^d)}
    I_2(\mu)
    =
    \frac{D^2}{2}.
    \label{eq:rectangular-quadratic-maximum}
\end{equation}
A probability measure is maximizing if and only if it is supported on a translate of \(H_2\) and every coordinate marginal is equally distributed between the two corresponding antipodal coordinate values.

\item[\textup{(iii)}]
If \(\alpha>2\), then
\begin{equation}
    \max_{\mu\in\mathcal P(T_{\mathbf L}^d)}
    I_\alpha(\mu)
    =
    \frac{D^\alpha}{2}.
    \label{eq:rectangular-superquadratic-maximum}
\end{equation}
The maximizers are precisely
\[
    \frac12
    \bigl(
        \delta_a+
        \delta_{a+(L_1,\ldots,L_d)}
    \bigr),
    \qquad
    a\in T_{\mathbf L}^d.
\]
\end{enumerate}
\end{theorem}

\subsection{Corner majorization}
\label{subsec:corner-majorization}

We begin with the range \(1\leq\alpha<2\). For
\[
    \mathbf a=(a_1,\ldots,a_d)
    \in
    \prod_{i=1}^d[0,L_i],
\]
set
\[
    F_\alpha(\mathbf a)
    =
    \left(
        \sum_{i=1}^d a_i^2
    \right)^{\alpha/2}.
\]
For fixed values of all variables except \(a_i\), differentiation away from the origin gives
\[
\begin{aligned}
    \frac{\partial^2 F_\alpha}{\partial a_i^2}
    =
    \alpha
    \left(
        \sum_{j=1}^d a_j^2
    \right)^{\alpha/2-2}
    \left(
        \sum_{j\neq i}a_j^2
        +
        (\alpha-1)a_i^2
    \right)
    \geq0.
\end{aligned}
\]
The limiting one-variable function at the origin is \(a_i^\alpha\), which is convex as well. Thus \(F_\alpha\) is separately convex.

Write
\[
    u_i=\frac{a_i}{L_i}.
\]
Successive interpolation between the endpoints of each coordinate gives
\begin{equation}
    F_\alpha(\mathbf a)
    \leq
    M_\alpha(\mathbf a),
    \label{eq:rectangular-corner-majorant}
\end{equation}
where
\[
\begin{aligned}
    M_\alpha(\mathbf a)
    =
    \sum_{\varepsilon\in\{0,1\}^d}
    F_\alpha(\varepsilon_1L_1,\ldots,\varepsilon_dL_d)
    \prod_{i=1}^d
    u_i^{\varepsilon_i}(1-u_i)^{1-\varepsilon_i}.
\end{aligned}
\]
Thus \(M_\alpha\) is the multilinear interpolant of the values of \(F_\alpha\) at the \(2^d\) corners.

The usefulness of this majorant comes from a positive-definite factorization. Put
\[
    p=\frac{\alpha}{2}\in\left[\frac12,1\right).
\]
The Bernstein representation
\[
    s^p
    =
    c_p
    \int_0^\infty
    \bigl(1-e^{-ts}\bigr)t^{-p-1}\,dt,
    \qquad
    c_p=\frac{p}{\Gamma(1-p)},
\]
is standard; see, for example, \cite{SchillingSongVondracek2012}. At a corner,
\[
\begin{aligned}
    F_\alpha(\varepsilon_1L_1,\ldots,\varepsilon_dL_d)
    =
    c_p
    \int_0^\infty
    \left[
        1-
        \prod_{i=1}^d e^{-t\varepsilon_iL_i^2}
    \right]
    t^{-p-1}\,dt.
\end{aligned}
\]
Since multilinear interpolation commutes with the integral, we obtain
\begin{equation}
    M_\alpha(\mathbf a)
    =
    c_p
    \int_0^\infty
    \left[
        1-
        \prod_{i=1}^d
        \bigl(1-\beta_i(t)a_i\bigr)
    \right]
    t^{-p-1}\,dt,
    \label{eq:bernstein-corner-interpolant}
\end{equation}
where
\[
    \beta_i(t)
    =
    \frac{1-e^{-tL_i^2}}{L_i}.
\]
In particular,
\[
    0<\beta_i(t)L_i<1.
\]

We next examine the Fourier coefficients of the one-dimensional factors. On \(\mathbb R/(2L_i\mathbb Z)\), index the characters by \(n\in\mathbb Z\):
\[
    x_i\longmapsto
    e^{\pi i n x_i/L_i}.
\]
For the circle distance
\[
    d_i(x_i,0)\in[0,L_i],
\]
a direct integration gives
\[
    \widehat d_i(0)=\frac{L_i}{2}
\]
and, for \(n\neq0\),
\begin{equation}
    \widehat d_i(n)
    =
    \begin{cases}
        0,
        & n\ \text{even},\\[1mm]
        -\dfrac{2L_i}{\pi^2n^2},
        & n\ \text{odd}.
    \end{cases}
    \label{eq:circle-distance-fourier}
\end{equation}

For \(t>0\), define
\[
    q_{i,t}(x_i)
    =
    1-\beta_i(t)d_i(x_i,0).
\]
Then
\[
    \widehat q_{i,t}(0)
    =
    1-\frac{\beta_i(t)L_i}{2}
    =
    \frac{1+e^{-tL_i^2}}{2}>0,
\]
while
\[
    \widehat q_{i,t}(n)
    =
    \begin{cases}
        0,
        & n\neq0\ \text{even},\\[1mm]
        \dfrac{2\beta_i(t)L_i}{\pi^2n^2},
        & n\ \text{odd}.
    \end{cases}
\]
Hence every Fourier coefficient of \(q_{i,t}\) is nonnegative.

Set
\[
    Q_t(x)
    =
    \prod_{i=1}^d q_{i,t}(x_i).
\]
Its Fourier coefficients are products of the one-dimensional coefficients and are therefore nonnegative. Moreover,
\begin{equation}
    \widehat Q_t(\mathbf n)>0
    \label{eq:rectangular-positive-frequency-set}
\end{equation}
whenever every \(n_i\) is either zero or odd and \(\mathbf n\neq0\).

Here and below we identify the dual lattice with \(\mathbb Z^d\) through
\[
    \mathbf n
    \longleftrightarrow
    \left(
        \frac{n_1}{2L_1},\ldots,
        \frac{n_d}{2L_d}
    \right).
\]
Since the coefficients in \eqref{eq:circle-distance-fourier} are absolutely summable, so are the Fourier coefficients of \(Q_t\). Thus \cref{prop:measure-energy-diagonalization} applies to every probability measure.

Let
\[
    \mathbf d(x,y)
    =
    \bigl(
        d_1(x_1,y_1),\ldots,d_d(x_d,y_d)
    \bigr).
\]
Using \eqref{eq:rectangular-corner-majorant} and then \eqref{eq:bernstein-corner-interpolant}, we obtain
\[
\begin{aligned}
    I_\alpha(\mu)
    &\leq
    \iint
    M_\alpha(\mathbf d(x,y))\,d\mu(x)\,d\mu(y)\\
    &=
    c_p
    \int_0^\infty
    \left[
        1-
        \iint Q_t(x-y)\,d\mu(x)\,d\mu(y)
    \right]
    t^{-p-1}\,dt.
\end{aligned}
\]
Since all Fourier coefficients of \(Q_t\) are nonnegative, \cref{cor:positive-definite-energy} gives
\[
    \iint Q_t(x-y)\,d\mu(x)\,d\mu(y)
    \geq
    \widehat Q_t(0).
\]
Consequently,
\begin{equation}
\begin{aligned}
    I_\alpha(\mu)
    &\leq
    c_p
    \int_0^\infty
    \bigl[1-\widehat Q_t(0)\bigr]
    t^{-p-1}\,dt\\
    &=
    2^{-d}
    \sum_{\varepsilon\in\{0,1\}^d}
    \left(
        \sum_{i=1}^d\varepsilon_iL_i^2
    \right)^{\alpha/2}.
\end{aligned}
    \label{eq:rectangular-global-upper-bound}
\end{equation}
Indeed, the last expression is simply the value of the multilinear function \(M_\alpha\) at
\[
    \left(
        \frac{L_1}{2},\ldots,\frac{L_d}{2}
    \right),
\]
and a multilinear function takes at the center of a box the average of its corner values.

The measure \(\mu_a\) attains \eqref{eq:rectangular-global-upper-bound}. If \(X,Y\) are independent with law \(\mu_a\), then for each \(i\),
\[
    d_i(X_i,Y_i)
    =
    \begin{cases}
        0,&\text{with probability }1/2,\\
        L_i,&\text{with probability }1/2,
    \end{cases}
\]
and these \(d\) choices are independent. This proves the value in \eqref{eq:rectangular-subquadratic-maximum}.

For later use, it is helpful to keep both nonnegative gaps in the preceding argument explicit. Write
\[
    B_\alpha
    =
    2^{-d}
    \sum_{\varepsilon\in\{0,1\}^d}
    \left(\sum_{i=1}^d \varepsilon_iL_i^2\right)^{\alpha/2},
\]
and set
\[
    \Delta_{\mathrm{corner}}(\mu)
    =
    \iint
    \bigl(M_\alpha(\mathbf d(x,y))-F_\alpha(\mathbf d(x,y))\bigr)
    \,d\mu(x)\,d\mu(y)
    \geq0
\]
and
\[
    \Delta_{\mathrm{spec}}(\mu)
    =
    c_p\int_0^\infty
    \left(
        \iint Q_t(x-y)\,d\mu(x)\,d\mu(y)
        -\widehat Q_t(0)
    \right)
    t^{-p-1}\,dt
    \geq0.
\]
The inequalities leading to \eqref{eq:rectangular-global-upper-bound} therefore give the exact deficit decomposition
\begin{equation}
    B_\alpha-I_\alpha(\mu)
    =
    \Delta_{\mathrm{corner}}(\mu)
    +
    \Delta_{\mathrm{spec}}(\mu).
    \label{eq:rectangular-deficit-decomposition}
\end{equation}
Consequently, an extremizer makes both deficits vanish.

The vanishing of the corner deficit gives the first equality condition. If \(1\leq\alpha<2\) and \(\mu\) attains \eqref{eq:rectangular-global-upper-bound}, then
\begin{equation}
    F_\alpha(\mathbf d(x,y))
    =
    M_\alpha(\mathbf d(x,y))
    \qquad
    \text{for all }x,y\in\operatorname{supp}\mu.
    \label{eq:corner-equality-on-support}
\end{equation}
Indeed, the difference in \eqref{eq:rectangular-corner-majorant} is continuous and nonnegative, and its integral against \(\mu\otimes\mu\) vanishes.

The vanishing of the spectral deficit gives the Fourier equality condition. By \cref{prop:measure-energy-diagonalization}, for every \(t>0\),
\[
    \iint Q_t(x-y)\,d\mu(x)\,d\mu(y)-\widehat Q_t(0)
    =
    \sum_{\mathbf n\neq0}
    \widehat Q_t(\mathbf n)
    |\widehat\mu(\mathbf n)|^2
    \geq0.
\]
Hence \(\Delta_{\mathrm{spec}}(\mu)=0\) implies that the series on the right is zero for almost every \(t>0\). For every \(\mathbf n=(n_1,\ldots,n_d)\neq0\) such that each \(n_i\) is either zero or odd, choose one \(t_0>0\) in this full-measure set. Every term in the series is nonnegative, and \eqref{eq:rectangular-positive-frequency-set} gives \(\widehat Q_{t_0}(\mathbf n)>0\). Therefore
\begin{equation}
    \widehat\mu(\mathbf n)=0.
    \label{eq:rectangular-fourier-equality}
\end{equation}
This proves the claimed Fourier equality condition.

\subsection{The subquadratic regime}
\label{subsec:rectangular-subquadratic}

We first treat \(1<\alpha<2\). In this range the separate convexity used above is strict. More precisely, for fixed values of the other coordinates,
\[
    s\longmapsto
    \left(
        s^2+\sum_{j\neq i}a_j^2
    \right)^{\alpha/2}
\]
is strictly convex on \([0,L_i]\). Hence equality in the corner interpolation is possible only at a corner. From \eqref{eq:corner-equality-on-support},
\begin{equation}
    d_i(x_i,y_i)\in\{0,L_i\}
    \qquad
    \text{for all }
    x,y\in\operatorname{supp}\mu
    \text{ and all }i.
    \label{eq:rectangular-endpoint-distances}
\end{equation}

Fix \(x^0\in\operatorname{supp}\mu\). It follows that
\[
    \operatorname{supp}\mu
    \subset
    x^0+H_2.
\]
Write
\[
    \mu
    =
    \sum_{\varepsilon\in\{0,1\}^d}
    w_\varepsilon\,
    \delta_{x^0+(\varepsilon_1L_1,\ldots,\varepsilon_dL_d)}.
\]
For
\[
    \chi\in\{0,1\}^d\setminus\{0\},
\]
apply \eqref{eq:rectangular-fourier-equality} to the frequency \(\mathbf n=\chi\). Up to a nonzero phase depending only on \(x^0\), this gives
\[
    \sum_{\varepsilon\in\{0,1\}^d}
    (-1)^{\chi\cdot\varepsilon}w_\varepsilon
    =
    0.
\]
Together with
\[
    \sum_\varepsilon w_\varepsilon=1,
\]
the Fourier inversion formula on the group \((\mathbb Z/2\mathbb Z)^d\) yields
\[
    w_\varepsilon=2^{-d}
\]
for every \(\varepsilon\). Thus \(\mu=\mu_{x^0}\).

At the endpoint \(\alpha=1\), the coordinatewise convexity used above is not always strict, so the equality set in \eqref{eq:rectangular-corner-majorant} must be identified more carefully.

\begin{lemma}[Equality in the corner interpolation at \(\alpha=1\)]
\label{lem:alpha-one-corner-equality}
Let
\[
    F_1(\mathbf a)=|\mathbf a|
\]
on \(\prod_i[0,L_i]\), and let \(M_1\) be its multilinear corner interpolant. Then
\[
    F_1(\mathbf a)=M_1(\mathbf a)
\]
if and only if one of the following holds:
\begin{enumerate}
\item[\textup{(a)}]
every \(a_i\) belongs to \(\{0,L_i\}\);

\item[\textup{(b)}]
there is exactly one index \(i\) for which \(a_i\in(0,L_i)\), and \(a_j=0\) for every \(j\neq i\).
\end{enumerate}
\end{lemma}

\begin{proof}
Let \(Z=(Z_1,\ldots,Z_d)\) have independent coordinates with
\[
    Z_i
    =
    \begin{cases}
        L_i,
        &\text{with probability }a_i/L_i,\\
        0,
        &\text{with probability }1-a_i/L_i.
    \end{cases}
\]
Then
\[
    \mathbb EZ=\mathbf a,
    \qquad
    M_1(\mathbf a)=\mathbb E|Z|.
\]
Thus
\[
    |\mathbf a|
    =
    |\mathbb EZ|
    \leq
    \mathbb E|Z|.
\]

Suppose first that \(\mathbb EZ\neq0\), and put
\[
    u=\frac{\mathbb EZ}{|\mathbb EZ|}.
\]
Then
\[
    |\mathbb EZ|
    =
    \mathbb E\langle u,Z\rangle
    \leq
    \mathbb E|Z|.
\]
Equality holds precisely when every nonzero value of \(Z\) having positive probability lies on the ray \(\mathbb R_+u\). If \(\mathbb EZ=0\), then \(\mathbf a=0\), which is already a corner.

Conditions \textup{(a)} and \textup{(b)} plainly give equality. Conversely, suppose that some \(a_i\in(0,L_i)\) and that \(a_j>0\) for some \(j\neq i\). By fixing \(Z_j=L_j\) and all other coordinates at values of positive probability, the two possibilities \(Z_i=0\) and \(Z_i=L_i\) produce two nonzero values of \(Z\) that are not collinear. Equality is therefore impossible. This exhausts all cases.
\end{proof}

The exceptional one-dimensional directions in \cref{lem:alpha-one-corner-equality} disappear once two support pairs are compared.

\begin{lemma}[Endpoint rigidity]
\label{lem:alpha-one-two-pair-rigidity}
Assume \(d\geq2\), and let \(\mu\) maximize \(I_1\) on \(T_{\mathbf L}^d\). Then
\[
    d_i(x_i,y_i)\in\{0,L_i\}
\]
for every \(x,y\in\operatorname{supp}\mu\) and every \(i\).
\end{lemma}

\begin{proof}
For each coordinate \(j\), let
\[
    \lambda_j=(\pi_j)_\#\mu
\]
be the \(j\)-th marginal. Taking in \eqref{eq:rectangular-fourier-equality} frequencies supported only in the \(j\)-th coordinate shows that
\[
    \widehat\lambda_j(n)=0
    \qquad
    \text{for every odd }n.
\]
If \(\tau_j(t)=t+L_j\), then
\[
    \widehat{(\tau_j)_\#\lambda_j}(n)
    =
    (-1)^n\widehat\lambda_j(n).
\]
Thus \(\lambda_j\) and \((\tau_j)_\#\lambda_j\) have the same Fourier coefficients, and Fourier--Stieltjes uniqueness on the circle gives
\begin{equation}
    (\tau_j)_\#\lambda_j=\lambda_j.
    \label{eq:half-period-marginal-invariance}
\end{equation}

Suppose, toward a contradiction, that there exist \(x^0,y\in\operatorname{supp}\mu\) and an index \(i\) with
\[
    a=d_i(x_i^0,y_i)\in(0,L_i).
\]
By \eqref{eq:corner-equality-on-support} and \cref{lem:alpha-one-corner-equality}, all other coordinate distances must vanish:
\begin{equation}
    d_k(x_k^0,y_k)=0,
    \qquad
    k\neq i.
    \label{eq:alpha-one-first-pair}
\end{equation}

Choose \(j\neq i\). Since \(x_j^0\in\operatorname{supp}\lambda_j\), the invariance \eqref{eq:half-period-marginal-invariance} implies
\[
    x_j^0+L_j
    \in
    \operatorname{supp}\lambda_j.
\]
Because the torus is compact,
\[
    \operatorname{supp}((\pi_j)_\#\mu)
    =
    \pi_j(\operatorname{supp}\mu),
\]
so there exists \(z\in\operatorname{supp}\mu\) with
\begin{equation}
    d_j(x_j^0,z_j)=L_j.
    \label{eq:alpha-one-antipodal-coordinate}
\end{equation}

Apply \cref{lem:alpha-one-corner-equality} to the pair \((x^0,z)\). Since its \(j\)-th coordinate distance is the nonzero endpoint \(L_j\), the \(i\)-th coordinate distance cannot lie in \((0,L_i)\). Hence
\[
    d_i(x_i^0,z_i)\in\{0,L_i\}.
\]
By \eqref{eq:alpha-one-first-pair}, \(y_j=x_j^0\), and therefore
\[
    d_j(y_j,z_j)=L_j.
\]
On the \(i\)-th circle,
\[
    d_i(y_i,z_i)
    \in
    \{a,L_i-a\}
    \subset(0,L_i).
\]
Thus the pair \((y,z)\) has one interior coordinate distance and a second nonzero coordinate distance, contradicting \cref{lem:alpha-one-corner-equality}. This proves the lemma.
\end{proof}

Fixing \(x^0\in\operatorname{supp}\mu\) and applying \cref{lem:alpha-one-two-pair-rigidity}, we again obtain
\[
    \operatorname{supp}\mu\subset x^0+H_2.
\]
The Fourier conditions \eqref{eq:rectangular-fourier-equality} then annihilate every nontrivial Walsh character, exactly as in the case \(1<\alpha<2\), and hence
\[
    \mu=\mu_{x^0}.
\]
This completes the proof of \cref{thm:rectangular-phase-diagram}\textup{(i)}.

The restriction \(d\geq2\) is essential at the endpoint. When \(d=1\) and \(\alpha=1\), the maximum is \(L_1/2\), but every probability measure invariant under the half-period translation \(x\mapsto x+L_1\) is maximizing.

\subsection{The quadratic and superquadratic regimes}
\label{subsec:rectangular-quadratic-superquadratic}

Let \(d_L\) denote the geodesic distance on \(\mathbb R/(2L\mathbb Z)\). From \eqref{eq:circle-distance-fourier}, every probability measure \(\lambda\) on this circle satisfies
\begin{equation}
    \iint d_L(x,y)\,d\lambda(x)\,d\lambda(y)
    =
    \frac{L}{2}
    -
    \frac{2L}{\pi^2}
    \sum_{\substack{n\in\mathbb Z\\ n\ \mathrm{odd}}}
    \frac{|\widehat\lambda(n)|^2}{n^2}
    \leq
    \frac{L}{2}.
    \label{eq:circle-mean-distance-bound}
\end{equation}
Equality holds if and only if all odd Fourier coefficients vanish, equivalently if and only if \(\lambda\) is invariant under the half-period translation.

Let \(X,Y\) be independent with law \(\mu\in\mathcal P(T_{\mathbf L}^d)\), and set
\[
    A_i=d_i(X_i,Y_i).
\]
Since \(0\leq A_i\leq L_i\),
\[
    A_i^2\leq L_iA_i.
\]
Applying \eqref{eq:circle-mean-distance-bound} to the \(i\)-th marginal of \(\mu\) gives
\begin{equation}
\begin{aligned}
    I_2(\mu)
    &=
    \sum_{i=1}^d\mathbb EA_i^2\\
    &\leq
    \sum_{i=1}^d L_i\mathbb EA_i\\
    &\leq
    \frac12\sum_{i=1}^dL_i^2
    =
    \frac{D^2}{2}.
\end{aligned}
    \label{eq:rectangular-quadratic-chain}
\end{equation}

We now characterize equality. If \(I_2(\mu)=D^2/2\), then
\[
    A_i(L_i-A_i)=0
    \qquad
    \mu\otimes\mu\text{-almost everywhere}
\]
for every \(i\). Since
\[
    (x,y)\longmapsto
    d_i(x_i,y_i)\bigl(L_i-d_i(x_i,y_i)\bigr)
\]
is continuous and nonnegative, it vanishes on \(\operatorname{supp}\mu\times\operatorname{supp}\mu\). Fixing \(x^0\in\operatorname{supp}\mu\), we obtain
\[
    \operatorname{supp}\mu\subset x^0+H_2.
\]
Equality in the second line of \eqref{eq:rectangular-quadratic-chain} also forces every coordinate marginal to be half-period invariant. Since each such marginal is now supported on two antipodal points, it assigns mass \(1/2\) to each.

Conversely, suppose that \(\mu\) is supported on a translate of \(H_2\) and each coordinate marginal is balanced between its two coordinate values. Then \(A_i\in\{0,L_i\}\) almost surely and
\[
    \mathbb P(A_i=L_i)=\frac12,
\]
so
\[
    \mathbb EA_i^2=\frac{L_i^2}{2}.
\]
Hence
\[
    I_2(\mu)=\frac{D^2}{2}.
\]
This proves \cref{thm:rectangular-phase-diagram}\textup{(ii)}. Notice that the coordinates need not be independent: the quadratic endpoint admits an entire family of balanced couplings on \(H_2\).

Finally let \(\alpha>2\). Since \(0\leq\rho(x,y)\leq D\),
\[
    \rho(x,y)^\alpha
    \leq
    D^{\alpha-2}\rho(x,y)^2.
\]
Together with \eqref{eq:rectangular-quadratic-maximum}, this gives
\[
    I_\alpha(\mu)
    \leq
    D^{\alpha-2}I_2(\mu)
    \leq
    \frac{D^\alpha}{2}.
\]
The equally weighted measure on a diametral antipodal pair attains this value.

Suppose now that equality holds. Equality in the pointwise estimate forces
\[
    \rho(x,y)\in\{0,D\}
    \qquad
    \text{for all }
    x,y\in\operatorname{supp}\mu,
\]
again by continuity. On an orthogonal rectangular torus, a point has a unique point at distance \(D\): each coordinate must simultaneously be at its antipode. Thus, for some \(a\),
\[
    \operatorname{supp}\mu
    \subset
    \bigl\{
        a,\,
        a+(L_1,\ldots,L_d)
    \bigr\}.
\]
Writing the two masses as \(s\) and \(1-s\), we obtain
\[
    I_\alpha(\mu)
    =
    2s(1-s)D^\alpha
    \leq
    \frac{D^\alpha}{2},
\]
with equality if and only if \(s=1/2\). This proves \cref{thm:rectangular-phase-diagram}\textup{(iii)}.

\section{The Regular Hexagonal Torus}
\label{sec:hexagonal-torus}

Unlike the rectangular case, the hexagonal problem is not separable. We use an explicit positive-definite function \(P\) lying below a geometric majorant \(h\), with contact set exactly a cyclic subgroup of order three.

Let
\[
    \Lambda_{\mathrm{hex}}
    =
    \mathbb Z b_1+\mathbb Z b_2,
    \qquad
    |b_1|=|b_2|=1,
    \qquad
    b_1\cdot b_2=\frac12,
\]
and set
\[
    T_{\mathrm{hex}}
    =
    \mathbb R^2/\Lambda_{\mathrm{hex}}.
\]
We use lattice coordinates
\[
    x=ub_1+vb_2,
    \qquad
    (u,v)\in\mathbb R^2/\mathbb Z^2.
\]
Let
\[
    q=\frac{b_1+b_2}{3},
    \qquad
    H_3=\{0,q,2q\}.
\]
Then \(H_3\) is a cyclic subgroup of order three. We denote normalized counting measure on \(H_3\) by
\[
    \nu_{H_3}
    =
    \frac13
    \bigl(
        \delta_0+\delta_q+\delta_{2q}
    \bigr).
\]

\begin{theorem}[Exact maximizers on the regular hexagonal torus]
\label{thm:hexagonal-extremizers}
For every \(\alpha\geq1\),
\begin{equation}
    \max_{\mu\in\mathcal P(T_{\mathrm{hex}})}
    I_\alpha(\mu)
    =
    \frac23\,3^{-\alpha/2}.
    \label{eq:hexagonal-energy-maximum}
\end{equation}
The maximizers are precisely the translates of \(\nu_{H_3}\).
\end{theorem}

\subsection{Symmetry reduction}
\label{subsec:hexagonal-symmetry}

In lattice coordinates the Euclidean quadratic form is
\[
    N(u,v)=u^2+uv+v^2.
\]
The Voronoi cell of \(\Lambda_{\mathrm{hex}}\) is a regular hexagon, and
\[
    D
    =
    \operatorname{diam}(T_{\mathrm{hex}})
    =
    \frac1{\sqrt3}.
\]
It is convenient to normalize the distance by setting
\[
    r(x)
    =
    \frac{\rho_{\Lambda_{\mathrm{hex}}}(x,0)}{D}
    =
    \sqrt3\,\rho_{\Lambda_{\mathrm{hex}}}(x,0).
\]
Thus \(0\leq r\leq1\).

By symmetry it suffices to work in the chamber
\begin{equation}
    \Delta
    =
    \bigl\{
        (u,v):
        0\leq v\leq u,\;
        2u+v\leq1
    \bigr\}.
    \label{eq:hexagonal-chamber}
\end{equation}
The triangle \(\Delta\) has vertices \((0,0)\), \((1/2,0)\), and \((1/3,1/3)\). Here \((1/3,1/3)\) is a Voronoi vertex and \((1/2,0)\) is the midpoint of an adjacent Voronoi edge. Joining the center of the regular hexagonal Voronoi cell to its six vertices and six edge midpoints partitions \(V\) into twelve congruent closed triangles with disjoint interiors; these are exactly the images of \(\Delta\) under the dihedral symmetry group. Their projections under the quotient map therefore tile \(T_{\mathrm{hex}}\) (up to their common boundaries). The nearest lattice point is the origin throughout \(\Delta\), so
\[
    r(u,v)^2
    =
    3(u^2+uv+v^2).
\]

Introduce the coordinates
\begin{equation}
    a=1-2u-v,
    \qquad
    b=u-v,
    \qquad
    c=u+2v.
    \label{eq:hexagonal-weyl-coordinates}
\end{equation}
On \(\Delta\),
\[
    a,b,c\geq0,
    \qquad
    a+b+c=1,
    \qquad
    b\leq c.
\]
A direct calculation gives
\begin{equation}
    r^2=b^2+bc+c^2.
    \label{eq:hexagonal-normalized-distance}
\end{equation}

Define on \(\Delta\)
\[
    h=a+\frac12bc.
\]
The chamber formula extends continuously across the walls. In lattice coordinates the three reflections in the sides of \(\Delta\) are
\[
\begin{aligned}
    \sigma_1(u,v)&=(u+v,-v), &&\text{across }v=0,\\
    \sigma_2(u,v)&=(v,u), &&\text{across }u=v,\\
    \sigma_3(u,v)&=(1-u-v,v)\pmod{\mathbb Z^2},
        &&\text{across }2u+v=1.
\end{aligned}
\]
Their linear parts preserve \(N(u,v)=u^2+uv+v^2\) and preserve \(\mathbb Z^2\), so they induce isometries of \(T_{\mathrm{hex}}\). Each \(\sigma_i\) fixes its corresponding side pointwise. Reflecting the formula \(h=a+bc/2\) from \(\Delta\) to the adjacent chambers therefore gives identical traces on every common side. Since the closed chambers tile the torus, the reflected definitions glue to a well-defined continuous function \(h\) on \(T_{\mathrm{hex}}\). Since
\[
    1-h
    =
    b+c-\frac12bc,
\]
we obtain
\begin{equation}
    (1-h)^2-r^2
    =
    \frac{bc}{4}(4a+bc)
    \geq0.
    \label{eq:hexagonal-geometric-majorization}
\end{equation}
In particular,
\begin{equation}
    0\leq r\leq1-h
    \qquad\text{on }T_{\mathrm{hex}}.
    \label{eq:r-below-one-minus-h}
\end{equation}

The subgroup \(H_3\) is adapted to this geometry. The nonzero differences between its points have normalized distance \(1\), while
\[
    h(0)=1,
    \qquad
    h(q)=h(2q)=0.
\]
Consequently,
\[
    \mathcal E_r(\nu_{H_3})=\frac23,
    \qquad
    \mathcal E_h(\nu_{H_3})=\frac13.
\]
In view of \eqref{eq:r-below-one-minus-h}, the case \(\alpha=1\) will follow once we prove
\[
    \mathcal E_h(\mu)\geq\frac13
    \qquad
    \text{for every }\mu\in\mathcal P(T_{\mathrm{hex}}),
\]
with equality only for translates of \(\nu_{H_3}\).

\subsection{A positive spectral minorant}
\label{subsec:hexagonal-spectral-minorant}

The certificate below is a continuous positive-Fourier minorant. Related linear-programming certificates have recently been used for fixed-cardinality torus energy problems, notably for Fibonacci lattices and tensor-product energies \cite{Nagel2025,BilykNagelRuohoniemi2026}. Those works concern fixed-cardinality tensor-product energies, whereas here the certificate is used to optimize over probability measures.

Let \(b_1^*,b_2^*\) be the basis dual to \(b_1,b_2\), and identify \(\Lambda_{\mathrm{hex}}^*\) with \(\mathbb Z^2\) by
\[
    (m,n)
    \longleftrightarrow
    mb_1^*+nb_2^*.
\]
Thus the frequency \((m,n)\) corresponds in lattice coordinates to the character
\[
    (u,v)\longmapsto e^{2\pi i(mu+nv)}.
\]

For \(n\geq0\), let
\[
    k_n=(n+1,-n)\in\mathbb Z^2.
\]
Let \(\mathcal O_n\) be its orbit under the dihedral symmetry group of \(\Lambda_{\mathrm{hex}}\), and set
\[
    C_n(u,v)
    =
    \frac1{|\mathcal O_n|}
    \sum_{k=(k_1,k_2)\in\mathcal O_n}
    \cos\bigl(2\pi(k_1u+k_2v)\bigr).
\]
Choose the weights
\[
    w_0
    =
    \frac6{\pi^2}-\frac1{12},
    \qquad
    w_n
    =
    \frac6{\pi^2(2n+1)^2},
    \qquad n\geq1.
\]
They are strictly positive and satisfy
\[
    \sum_{n\geq0}w_n=\frac23.
\]
Define
\begin{equation}
    P(x)
    =
    \frac13+\sum_{n\geq0}w_nC_n(x).
    \label{eq:hexagonal-spectral-minorant}
\end{equation}
The series converges absolutely and uniformly. Since each \(C_n\) has nonnegative Fourier coefficients, \(P\) is positive definite, and
\[
    \widehat P(0)=\frac13.
\]

We next put \(P\) into a form suitable for pointwise comparison with \(h\). Pairing opposite frequencies in \(\mathcal O_n\) gives, on \(\Delta\),
\[
\begin{aligned}
    C_n(u,v)
    =
    \frac13\bigl\{&
    \cos\pi(u+v)\cos((2n+1)\pi b)\\
    &+
    \cos\pi u\cos((2n+1)\pi c)\\
    &+
    \cos\pi v\cos((2n+1)\pi(b+c))
    \bigr\}.
\end{aligned}
\]
Let
\[
    W(t)
    =
    \sum_{n\geq0}
    w_n\cos((2n+1)\pi t),
    \qquad
    0\leq t\leq1.
\]
The classical odd-cosine identity
\[
    \sum_{n\geq0}
    \frac{\cos((2n+1)\pi t)}{(2n+1)^2}
    =
    \frac{\pi^2}{8}(1-2t)
\]
yields
\begin{equation}
    W(t)
    =
    \frac34(1-2t)-\frac1{12}\cos\pi t.
    \label{eq:hexagonal-odd-cosine-kernel}
\end{equation}

Set
\[
    s=b+c,
    \qquad
    d=c-b.
\]
Then
\[
    a=1-s,
    \qquad
    b=\frac{s-d}{2},
    \qquad
    c=\frac{s+d}{2},
\]
and
\begin{equation}
    h
    =
    1-s+\frac{s^2-d^2}{8}.
    \label{eq:hexagonal-h-sd}
\end{equation}
Moreover,
\[
    u=\frac{s}{2}-\frac{d}{6},
    \qquad
    v=\frac{d}{3},
    \qquad
    u+v=\frac{s}{2}+\frac{d}{6}.
\]
Write
\[
    A
    =
    \cos\pi\left(\frac{s}{2}+\frac{d}{6}\right),
    \qquad
    B
    =
    \cos\pi\left(\frac{s}{2}-\frac{d}{6}\right),
    \qquad
    C
    =
    \cos\frac{\pi d}{3},
\]
and set
\[
    L_0
    =
    (1-s+d)A
    +(1-s-d)B
    +(1-2s)C.
\]
Finally, define
\[
    \Theta_1(u,v)
    =
    \frac29
    \bigl[
        3
        -\cos2\pi u
        -\cos2\pi v
        -\cos2\pi(u+v)
    \bigr].
\]
Using
\[
\begin{aligned}
    &A\cos\pi b
    +B\cos\pi c
    +C\cos\pi s\\
    &\hspace{18mm}
    =
    \cos2\pi u
    +\cos2\pi v
    +\cos2\pi(u+v),
\end{aligned}
\]
and substituting \eqref{eq:hexagonal-odd-cosine-kernel} into \eqref{eq:hexagonal-spectral-minorant}, we obtain
\begin{equation}
    P
    =
    \frac14(1+L_0)
    +
    \frac18\Theta_1.
    \label{eq:hexagonal-P-closed-form}
\end{equation}

The central estimate is the following.

\begin{proposition}[Pointwise spectral minorant]
\label{prop:hexagonal-spectral-minorant}
On \(T_{\mathrm{hex}}\),
\begin{equation}
    P(x)\leq h(x).
    \label{eq:hexagonal-P-below-h}
\end{equation}
Moreover,
\[
    P(x)=h(x)
    \quad\Longleftrightarrow\quad
    x\in H_3.
\]
\end{proposition}

\begin{proof}
By symmetry it suffices to work on \(\Delta\). Put
\[
    \mathcal D=h-P.
\]
If
\[
    S
    =
    A\cos\pi b+B\cos\pi c+C\cos\pi s,
\]
then the definition of \(\Theta_1\) and \eqref{eq:hexagonal-P-closed-form} give
\begin{equation}
    \mathcal D
    =
    h-\frac13-\frac14L_0+\frac1{36}S.
    \label{eq:hexagonal-D-closed-form}
\end{equation}
We divide the chamber into three regions, using different elementary trigonometric bounds in each.

\subsubsection*{Region I: \(a\geq\frac12\)}

Here
\[
    0\leq d\leq s\leq\frac12.
\]
In particular,
\[
    1-s+d\geq0,
    \qquad
    1-s-d\geq0,
    \qquad
    1-2s\geq0.
\]
In the term \(-L_0/4\) of \eqref{eq:hexagonal-D-closed-form} we therefore use the upper bound
\[
    \cos(\pi t)\leq1-4t^2,
    \qquad
    0\leq t\leq\frac12,
\]
at
\[
    t=\frac{s}{2}+\frac{d}{6},
    \qquad
    t=\frac{s}{2}-\frac{d}{6},
    \qquad
    t=\frac{d}{3}.
\]
For the three products in \(S\) we use
\[
    \cos x\geq1-\frac{x^2}{2}.
\]
The quadratic lower bound for the first factor in each product is positive on the present region; if the corresponding lower bound for the second factor is negative, the product estimate is automatic since the actual product is nonnegative. More explicitly,
\[
\begin{aligned}
    L_0
    &\leq
    (1-s+d)\left(1-\frac{(3s+d)^2}{9}\right)
    +(1-s-d)\left(1-\frac{(3s-d)^2}{9}\right)\\
    &\qquad
    +(1-2s)\left(1-\frac{4d^2}{9}\right),
\end{aligned}
\]
and
\[
\begin{aligned}
    S
    &\geq
    \left(1-\frac{\pi^2(3s+d)^2}{72}\right)
    \left(1-\frac{\pi^2(s-d)^2}{8}\right)\\
    &\quad+
    \left(1-\frac{\pi^2(3s-d)^2}{72}\right)
    \left(1-\frac{\pi^2(s+d)^2}{8}\right)\\
    &\quad+
    \left(1-\frac{\pi^2d^2}{18}\right)
    \left(1-\frac{\pi^2s^2}{2}\right).
\end{aligned}
\]
Substituting these two bounds into \eqref{eq:hexagonal-D-closed-form} and collecting terms gives
\begin{equation}
\begin{aligned}
    \mathcal D
    \geq{}&
    \frac1{648}R_1(s,d)\\
    &+
    \frac{\pi^4}{10368}
    (s^2-d^2)(9s^2-d^2)
    \geq
    \frac1{648}R_1(s,d),
\end{aligned}
    \label{eq:hexagonal-region-one-bound}
\end{equation}
where
\[
\begin{aligned}
    R_1(s,d)
    ={}&
    s^2
    \bigl(
        405-18\pi^2-324s
    \bigr)\\
    &+
    d^2
    \bigl(
        \pi^4s^2+36s+27-6\pi^2
    \bigr).
\end{aligned}
\]
The first coefficient is positive on \(0\leq s\leq1/2\), since
\[
    405-18\pi^2-324s
    \geq
    243-18\pi^2>0.
\]
If the coefficient of \(d^2\) is nonnegative, then \(R_1\geq0\), with equality only if \(s=d=0\). If that coefficient is negative, the inequality \(d^2\leq s^2\) gives
\[
    R_1(s,d)
    \geq
    s^2
    \bigl(
        432-24\pi^2-288s+\pi^4s^2
    \bigr).
\]
The quadratic factor is decreasing on \([0,1/2]\), because
\[
    2\pi^4s-288<0
\]
there, and at \(s=1/2\) it equals
\[
    288-24\pi^2+\frac{\pi^4}{4}>0.
\]
Thus
\[
    \mathcal D\geq0
\]
throughout Region I, with equality only at
\[
    s=d=0,
\]
that is, at the origin.

\subsubsection*{Region II: \(a\leq\frac12\) and \(c\leq\frac12\)}

The chamber constraints imply
\[
    0\leq d\leq a\leq\frac12.
\]
Since \(s=1-a\), the trigonometric factors in
\eqref{eq:hexagonal-D-closed-form} become
\[
\begin{aligned}
    A&=\sin\frac{\pi(3a-d)}6,
    &
    B&=\sin\frac{\pi(3a+d)}6,
    &
    C&=\cos\frac{\pi d}{3},\\
    \cos\pi b&=\sin\frac{\pi(a+d)}2,
    &
    \cos\pi c&=\sin\frac{\pi(a-d)}2,
    &
    \cos\pi s&=-\cos\pi a.
\end{aligned}
\]
All sine arguments lie in \([0,\pi/2]\). We use
\[
    \frac{2x}{\pi}\leq\sin x\leq x,
    \qquad
    0\leq x\leq\frac{\pi}{2},
\]
together with
\[
    \cos x\geq1-\frac{x^2}{2}
\]
and
\[
    \cos(\pi t)\leq1-4t^2,
    \qquad
    0\leq t\leq\frac12.
\]
In \(-L_0/4\), the coefficients of \(A\) and \(B\) are nonpositive,
whereas the coefficient of \(C\) is nonnegative. Thus we use the upper
sine bounds for \(A,B\) and the quadratic lower bound for \(C\). In
\(S/36\), the first two products are bounded below using the chord
estimate for each sine. For the last product, \(0\leq C\leq1\) and
\[
    -C\cos\pi a
    \geq
    -\cos\pi a
    \geq
    -(1-4a^2).
\]
Thus
\[
\begin{aligned}
    -\frac14L_0
    &\geq
    -\frac{\pi}{24}
    \bigl[(a+d)(3a-d)+(a-d)(3a+d)\bigr]\\
    &\qquad
    +\frac{1-2a}{4}
    \left(1-\frac{\pi^2d^2}{18}\right),
\end{aligned}
\]
while
\[
    \frac1{36}S
    \geq
    \frac1{36}
    \left[
        \frac{(3a-d)(a+d)+(3a+d)(a-d)}{3}
        -(1-4a^2)
    \right].
\]
Together with
\[
    h=a+\frac{(1-a)^2-d^2}{8},
\]
collecting terms in \eqref{eq:hexagonal-D-closed-form} gives
\begin{equation}
    216\mathcal D
    \geq
    N(a,d),
    \label{eq:hexagonal-region-two-bound}
\end{equation}
where
\[
\begin{aligned}
    N(a,d)
    ={}&
    3+54a+(63-54\pi)a^2\\
    &+
    \bigl(
        6\pi^2a+18\pi-3\pi^2-31
    \bigr)d^2.
\end{aligned}
\]
The coefficient of \(d^2\) is bounded below by its value at \(a=0\),
and
\[
    18\pi-3\pi^2-31>-5.
\]
For instance, the latter follows from \(3<\pi<22/7\). Since
\(d^2\leq a^2\),
\[
    N(a,d)
    \geq
    3+54a+(58-54\pi)a^2.
\]
The quadratic on the right is concave on \([0,1/2]\), so its minimum
there is attained at an endpoint. Its endpoint values are
\[
    3
    \qquad\text{and}\qquad
    \frac{89-27\pi}{2}>0,
\]
the last inequality again following from \(\pi<22/7\). Hence
\[
    \mathcal D>0
\]
throughout Region II.

\subsubsection*{Region III: \(c\geq\frac12\)}

Set
\[
    t=c=1-a-b
\]
and keep \(b\) fixed. The parameter range is
\begin{equation}
    0\leq b\leq\frac12,
    \qquad
    \frac12\leq t\leq1-b.
    \label{eq:hexagonal-region-three-domain}
\end{equation}
Define
\[
    x=\frac{\pi(t-b)}3,
    \qquad
    y=\frac{\pi(2b+t)}3,
    \qquad
    z=x+y.
\]
Then
\[
    0\leq x\leq y\leq\frac\pi2,
    \qquad
    \frac\pi3\leq z\leq\frac{2\pi}3.
\]

Differentiating \(\mathcal D=h-P\) twice with respect to \(t\), using \eqref{eq:hexagonal-P-closed-form}, gives
\begin{align}
    \mathcal D_{tt}
    ={}&
    \frac{\pi^2}{9}(1-2b)\cos z
    -
    \frac\pi3(\sin x+\sin y)
    \notag\\
    &-
    \frac{\pi^2}{9}
    \left[
        \left(\frac t2-\frac14\right)\cos y
        +
        \left(\frac{b+t}{2}-\frac14\right)\cos x
    \right]
    \notag\\
    &-
    \frac{\pi^2}{81}
    \bigl[
        \cos2x+\cos2y+4\cos2z
    \bigr].
    \label{eq:hexagonal-second-derivative}
\end{align}

The range \eqref{eq:hexagonal-region-three-domain} gives
\[
\begin{aligned}
    \sin x+\sin y
    &=
    2
    \sin\frac{\pi(2t+b)}6
    \cos\frac{\pi b}{2}\\
    &\geq
    2
    \sin\frac{\pi(1+b)}6
    \cos\frac{\pi b}{2}.
\end{aligned}
\]
Set
\[
    G(b)
    =
    2
    \sin\frac{\pi(1+b)}6
    \cos\frac{\pi b}{2}.
\]
The addition formula gives
\[
    G(b)
    =
    \sin\frac{\pi(1+4b)}6
    +
    \sin\frac{\pi(1-2b)}6.
\]
Both terms are concave on \(0\leq b\leq1/2\), so \(G\) is concave. Since
\[
    G(0)=G(1/2)=1,
\]
we obtain
\begin{equation}
    \sin x+\sin y\geq1.
    \label{eq:hexagonal-sine-lower-bound}
\end{equation}

Next set
\[
    S_2
    =
    \cos2x+\cos2y+4\cos2z.
\]
Using
\[
    \cos2x+\cos2y
    =
    2\cos z\cos(y-x),
\]
we have
\[
    S_2
    =
    8\cos^2z
    +
    2\cos z\cos(y-x)
    -4.
\]
If \(z\leq\pi/2\), then \(\cos z\geq0\). Moreover
\[
    y-x=\pi b\in[0,\pi/2],
\]
so \(\cos(y-x)\geq0\). The first two terms in the last expression for \(S_2\) are therefore nonnegative, and hence
\[
    S_2\geq-4.
\]
If \(z\geq\pi/2\), write
\[
    q_0=-\cos z\in[0,1/2].
\]
Since \(\cos(y-x)\leq1\),
\[
    S_2
    \geq
    8q_0^2-2q_0-4
    \geq
    -\frac{33}{8}.
\]

The two coefficients in the square brackets of \eqref{eq:hexagonal-second-derivative} are nonnegative by \eqref{eq:hexagonal-region-three-domain}. If \(z\leq\pi/2\), then \(z\in[\pi/3,\pi/2]\), so \(0\leq\cos z\leq1/2\); moreover \(0\leq1-2b\leq1\). Hence
\[
    (1-2b)\cos z\leq\frac12,
\]
and therefore
\[
    \mathcal D_{tt}
    \leq
    -\frac\pi3
    +
    \frac{17\pi^2}{162}
    <0,
\]
since \(17\pi<54\).
If \(z\geq\pi/2\), the first term in \eqref{eq:hexagonal-second-derivative} is nonpositive, and therefore
\[
    \mathcal D_{tt}
    \leq
    -\frac\pi3
    +
    \frac{11\pi^2}{216}
    <0,
\]
since \(11\pi<72\).
Thus, for fixed \(b\), the function \(t\mapsto\mathcal D(b,t)\) is strictly concave. Its minimum on \([1/2,1-b]\) is therefore attained at an endpoint.

The endpoint \(t=1/2\) lies in Region II and gives a strictly positive value. It remains to treat
\[
    t=1-b,
    \qquad\text{equivalently }a=0.
\]
On this boundary let
\[
    d=c-b\in[0,1].
\]
Then
\[
    \mathcal D=E(d),
\]
where
\begin{equation}
\begin{aligned}
    E(d)
    ={}&
    -\frac{d^2}{8}
    +
    \frac d2\sin\frac{\pi d}{6}
    +
    \frac7{36}\cos\frac{\pi d}{3}\\
    &+
    \frac1{36}\cos\frac{2\pi d}{3}
    -
    \frac5{24}.
\end{aligned}
    \label{eq:hexagonal-boundary-function}
\end{equation}
Put
\[
    \theta=\frac{\pi d}{6}.
\]
Differentiating \eqref{eq:hexagonal-boundary-function} gives
\[
\begin{aligned}
    108E'(d)
    ={}&
    54\theta\cos\theta
    -\frac{162}{\pi}\theta
    +54\sin\theta\\
    &-
    7\pi\sin2\theta
    -
    2\pi\sin4\theta
    =:F_0(\theta).
\end{aligned}
\]
One checks that
\[
    F_0(0)=F_0(\pi/6)=0,
\]
while
\[
\begin{aligned}
    F_0''(\theta)
    ={}&
    \sin\theta
    \bigl[
        -162
        +56\pi\cos\theta
        +128\pi\cos\theta\cos2\theta
    \bigr]\\
    &-
    54\theta\cos\theta.
\end{aligned}
\]
For \(0\leq\theta\leq\pi/6\),
\[
    \cos\theta\geq\frac{\sqrt3}{2},
    \qquad
    \cos2\theta\geq\frac12,
    \qquad
    \theta\cos\theta\leq\sin\theta.
\]
Therefore
\[
    F_0''(\theta)
    \geq
    (60\sqrt3\,\pi-216)\sin\theta
    \geq0.
\]
Here \(60\sqrt3\,\pi-216>0\) (equivalently \(\sqrt3\,\pi>18/5\)). The inequality is strict for \(\theta>0\), so \(F_0\) is strictly convex on \([0,\pi/6]\). Since its endpoint values are zero,
\[
    F_0(\theta)<0,
    \qquad
    0<\theta<\frac{\pi}{6}.
\]
Thus
\[
    E'(d)<0,
    \qquad
    0<d<1.
\]
Finally,
\[
    E(1)=0,
\]
and hence
\[
    E(d)>0
    \qquad
    \text{for }0\leq d<1.
\]
The only zero in Region III is therefore
\[
    a=0,
    \qquad
    b=0,
    \qquad
    c=1.
\]

Combining the three regions proves \eqref{eq:hexagonal-P-below-h}. In the chosen chamber the contact set consists of the origin and
\[
    (u,v)=\left(\frac13,\frac13\right).
\]
The six dihedral images of this point are the six Voronoi vertices. In lattice coordinates, modulo \(\mathbb Z^2\), they fall into exactly two classes,
\[
    q=\left(\frac13,\frac13\right)
    \qquad\text{and}\qquad
    -q\equiv\left(\frac23,\frac23\right)=2q.
\]
Hence the global contact set is exactly \(H_3=\{0,q,2q\}\).
\end{proof}

\subsection{The extremal measure}
\label{subsec:hexagonal-extremal-measure}

Since the series in \eqref{eq:hexagonal-spectral-minorant} is absolutely convergent, \cref{prop:measure-energy-diagonalization} gives, for every \(\mu\in\mathcal P(T_{\mathrm{hex}})\),
\begin{align}
    \mathcal E_P(\mu)
    &=
    \frac13
    +
    \sum_{n\geq0}
    \frac{w_n}{|\mathcal O_n|}
    \sum_{k\in\mathcal O_n}
    |\widehat\mu(k)|^2
    \notag\\
    &\geq
    \frac13.
    \label{eq:hexagonal-positive-energy-bound}
\end{align}
Together with \cref{prop:hexagonal-spectral-minorant}, this yields
\[
    \mathcal E_h(\mu)
    \geq
    \mathcal E_P(\mu)
    \geq
    \frac13.
\]
Since
\[
    \mathcal E_h(\nu_{H_3})=\frac13,
\]
we obtain
\begin{equation}
    \min_{\mu\in\mathcal P(T_{\mathrm{hex}})}
    \mathcal E_h(\mu)
    =
    \frac13.
    \label{eq:hexagonal-h-minimum}
\end{equation}

Suppose now that
\[
    \mathcal E_h(\mu)=\frac13.
\]
Then equality holds throughout the preceding chain, so
\[
    \mathcal E_{h-P}(\mu)=0.
\]
The function \(h-P\) is continuous and nonnegative. Hence
\[
    (h-P)(x-y)=0
    \qquad
    \text{for all }
    x,y\in\operatorname{supp}\mu.
\]
By \cref{prop:hexagonal-spectral-minorant},
\[
    x-y\in H_3
    \qquad
    \text{for all }
    x,y\in\operatorname{supp}\mu.
\]
Fixing \(x_0\in\operatorname{supp}\mu\), we conclude that
\[
    \operatorname{supp}\mu
    \subset
    x_0+H_3.
\]

Translate the measure so that \(x_0=0\), and write
\[
    \mu
    =
    \lambda_0\delta_0
    +
    \lambda_1\delta_q
    +
    \lambda_2\delta_{2q},
    \qquad
    \lambda_0+\lambda_1+\lambda_2=1.
\]
Equality in \eqref{eq:hexagonal-positive-energy-bound} and the strict positivity of \(w_0\) give
\[
    \widehat\mu(k)=0
    \qquad
    \text{for every }k\in\mathcal O_0.
\]
Taking \(k_0=(1,0)\) and writing
\[
    \omega=e^{-2\pi i/3},
\]
we obtain
\[
    0
    =
    \widehat\mu(k_0)
    =
    \lambda_0+\lambda_1\omega+\lambda_2\omega^2.
\]
Together with
\[
    \lambda_0+\lambda_1+\lambda_2=1,
\]
this forces
\[
    \lambda_0=\lambda_1=\lambda_2=\frac13.
\]
Thus the minimizers in \eqref{eq:hexagonal-h-minimum} are precisely the translates of \(\nu_{H_3}\).

We now return to the distance kernel. From \eqref{eq:r-below-one-minus-h},
\[
\begin{aligned}
    \mathcal E_r(\mu)
    &\leq
    1-\mathcal E_h(\mu)\\
    &\leq
    \frac23.
\end{aligned}
\]
The measure \(\nu_{H_3}\) attains \(2/3\), so
\[
    \max_{\mu\in\mathcal P(T_{\mathrm{hex}})}
    \mathcal E_r(\mu)
    =
    \frac23.
\]
Moreover, equality forces \(\mathcal E_h(\mu)=1/3\), and hence the maximizers are precisely the translates of \(\nu_{H_3}\).

Finally, for every \(\alpha\geq1\),
\[
    0\leq r^\alpha\leq r
\]
because \(0\leq r\leq1\). Therefore
\[
    \mathcal E_{r^\alpha}(\mu)
    \leq
    \mathcal E_r(\mu)
    \leq
    \frac23.
\]
For \(\nu_{H_3}\), every pairwise normalized distance is either \(0\) or \(1\), so
\[
    \mathcal E_{r^\alpha}(\nu_{H_3})
    =
    \frac23
\]
for every \(\alpha\geq1\). If equality holds for some \(\mu\), then
\[
    \frac23
    =
    \mathcal E_{r^\alpha}(\mu)
    \leq
    \mathcal E_r(\mu)
    \leq
    \frac23,
\]
so \(\mu\) is again a translate of \(\nu_{H_3}\).

Since
\[
    \rho_{\Lambda_{\mathrm{hex}}}
    =
    D r,
    \qquad
    D=\frac1{\sqrt3},
\]
we conclude that
\[
    I_\alpha(\mu)
    =
    3^{-\alpha/2}
    \mathcal E_{r^\alpha}(\mu),
\]
and hence
\[
    \max_{\mu\in\mathcal P(T_{\mathrm{hex}})}
    I_\alpha(\mu)
    =
    \frac23\,3^{-\alpha/2}.
\]
This proves \cref{thm:hexagonal-extremizers}.

\section{Concluding Remarks}
\label{sec:concluding-remarks}

For flat tori of dimension at least two, the cut locus leaves a quantitative Fourier signature in every positive power of the distance: each Voronoi facet produces positive high-frequency coefficients along frequencies approaching its normal, with a leading term determined by the facet. These coefficients yield finite negative-type obstructions and smooth perturbations that increase the distance energy away from Haar measure. The exact maximization results show, in contrast, that the global extremizers depend sensitively on the lattice geometry: rectangular tori exhibit a transition at \(\alpha=2\), whereas the regular hexagonal torus has the same three-point extremizers throughout \(\alpha\geq1\).

Several natural questions remain. For an arbitrary two-dimensional lattice, one may ask for geometric conditions on the Voronoi cell that force a maximizing measure to be supported on a finite subgroup, and whether the order and geometry of that subgroup can be predicted from the lattice.

The Fourier asymptotics also leave open the complete high-frequency sign pattern of
\[
    \widehat{\rho_\Lambda^\alpha}(\xi).
\]
It would be interesting to determine how much of the Voronoi geometry is encoded by these signs and asymptotics. In the hexagonal case, one can further ask whether the infinite positive spectral chain used in \cref{subsec:hexagonal-spectral-minorant} admits a finite positive-definite certificate with the same contact set.

The exact maximization problem for \(0<\alpha<1\) remains open, already for the square torus. The arguments used here are no longer sharp in that regime, and a different analysis appears to be required.

\section*{Acknowledgement}
The author used OpenAI ChatGPT (GPT-5.6 Sol) to discuss proof strategies, check calculations, and assist with manuscript preparation. All mathematical statements, proofs, computations, and citations were independently verified by the author, who takes full responsibility for the content of the manuscript.

\raggedbottom
\bibliographystyle{amsalpha}
\bibliography{references}

\end{document}